\documentclass[journal,12pt,onecolumn,draftclsnofoot]{IEEEtran}
\usepackage{amsmath,amssymb,amsthm,mathrsfs,mathtools,booktabs, array,microtype,graphicx}
\usepackage{xcolor}
\usepackage{tikz,pgfplots}
\usetikzlibrary{arrows.meta,positioning,calc,fit,backgrounds,patterns}
\usepgfplotslibrary{groupplots,fillbetween}
\definecolor{figBlue}{HTML}{216A87}
\definecolor{figOrange}{HTML}{B75B30}
\definecolor{figTeal}{HTML}{2B8C7F}
\definecolor{figGray}{HTML}{6B7280}
\pgfplotsset{compat=1.18}
\pgfplotsset{
  paper axis/.style={
    /tikz/font=\footnotesize,
    tick label style={font=\footnotesize},
    label style={font=\footnotesize},
    title style={font=\footnotesize},
    legend style={font=\footnotesize,draw=none,fill=none},
    axis line style={black!65,line width=.45pt},
    tick style={black!65,line width=.4pt},
    major grid style={black!9,line width=.3pt},
    minor grid style={black!5,line width=.25pt},
    tick align=outside,
    scaled ticks=false,
    clip=true,
    every axis plot/.append style={line width=1pt}
  }
}
\tikzset{
  paper arrow/.style={-{Stealth[length=2mm,width=1.4mm]},line width=.75pt},
  paper block/.style={draw=figBlue,line width=.7pt,rounded corners=1.5pt,fill=figBlue!5,align=center,font=\footnotesize},
  paper note/.style={font=\footnotesize,text=figGray,align=center}
}
\usepackage{placeins,balance,adjustbox}
\usepackage[hidelinks]{hyperref}
\newcommand{\paperfigure}[1]{%
  \begin{adjustbox}{max width=\linewidth}#1\end{adjustbox}}

\newtheorem{theorem}{Theorem}
\newtheorem{lemma}{Lemma}
\newtheorem{proposition}{Proposition}

\newtheorem{remark}{Remark}

\newcommand{\enc}{\mathscr{E}}\newcommand{\dec}{\mathscr{D}}\newcommand{\R}{\mathbb R}\newcommand{\C}{\mathbb C}
\newcommand{\T}{\mathsf T}
\renewcommand{\Re}{\mathop{\mathfrak{Re}}\nolimits}
\renewcommand{\Im}{\mathop{\mathfrak{Im}}\nolimits}

\title{It Takes Two Moments: \\ Rigidity and Limits of Continuous Encoders for \\Over-the-Air Computation}
\author{Saeed Razavikia,~\IEEEmembership{Senior Member,~IEEE} \thanks{Saeed Razavikia is with the Department of Network and Systems Engineering, KTH Royal Institute of Technology, Stockholm, Sweden (e-mail: \href{mailto:sraz@kth.se}{sraz@kth.se}).}}
\begin{document}

\maketitle

\begin{abstract}
We study over-the-air computation of functions, such as empirical variance, formed by combining two separately postprocessed sums of locally transformed observations. We consider continuous encoders subject to pointwise peak-power constraints over a single use of a complex multiple-access channel. With four or more transmitters, we identify conditions under which exact noiseless recovery forces every continuous encoder collection into a common affine structure that preserves both sums, reducing infinite-dimensional encoder design to selecting a finite set of coefficients. To quantify resilience against channel noise, we introduce a robustness score that measures the separation of noiseless received aggregates relative to variations in the desired function value. Within this characterized encoder family, we formulate robustness maximization as an exact convex semi-infinite program, eliminating rank constraints without a relaxation gap, and construct a decoder that attains the minimum uniform noise-amplification constant. For several function classes, including empirical variance, we derive closed-form robustness-optimal encoders and scores. Under mild regularity conditions, we also establish sharp limits on the Gaussian mean squared error at high signal-to-noise ratios and show how using the input distribution reduces average error at the expense of uniform robustness.
\end{abstract}


\section{Introduction}
\IEEEPARstart{O}{ver}-the-air computation harnesses signal superposition over a multiple-access channel to evaluate functions of distributed observations directly, avoiding the need to decode individual messages~\cite{sahin2023survey}. In large networks, concurrent transmission substantially improves spectral efficiency and reduces latency relative to orthogonal access. Specifically, each transmitter $k$ maps its local observation $s_k$ through an encoder $\mathscr E_k$ to a complex symbol. In a single complex channel use, the receiver observes
$$
y=\sum\nolimits_{k=1}^{K}\mathscr E_k(s_k)+z,
$$
where $z$ denotes additive noise. The receiver uses $y$ to estimate the desired function $f(\boldsymbol s)$. The central design challenge is to choose the encoders $\mathscr E_k$ and decoder so that the received superposition preserves the task-relevant information and permits reliable estimation in noise.

Computing sums and linear combinations naturally exploits the additive channel structure~\cite{nazer2007}. Through local preprocessing and receiver postprocessing, analog over-the-air computation extends this principle to \emph{nomographic} functions of real-valued inputs~\cite{golden2013harnessing}, supporting applications in federated learning, distributed inference, and wireless control~\cite{amiri2020federated,yilmaz2025private,park2021optimized}. However, a single real aggregate with continuous processing cannot represent every continuous multivariate function. For empirical variance with $K\ge4$, Theorem~\ref{thm:rigidity} shows that exact computation requires both real dimensions of the complex aggregate. For finite input alphabets, ChannelComp overcomes this limitation by assigning complex constellation points to local inputs so that received sums corresponding to distinct function values remain well separated~\cite{channelcomp2024}. While this distance-based principle enables exact noiseless recovery of arbitrary functions over finite alphabets, discrete constellation design faces a combinatorial explosion in the number of input configurations as the alphabet or network size grows~\cite{sumcomp2025,pyramidsampling2025}.

%-------------------------
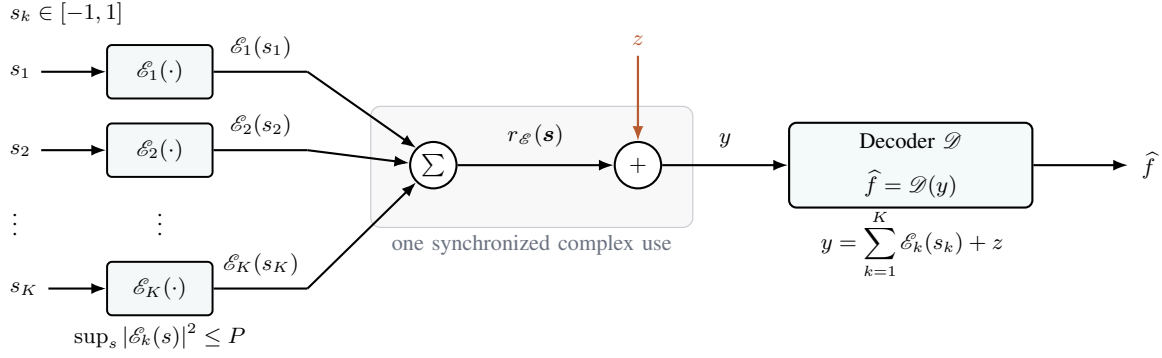
\begin{figure*}[t]
 \centering\paperfigure{%
% BEGIN figures/tikz/system_model
% Native TikZ: one synchronized complex multiple-access channel use.
\begin{tikzpicture}[
  x=1cm,y=.82cm,>=Latex,font=\footnotesize,
  signal/.style={draw=black,line width=.85pt,-{Latex[length=2.1mm]}},
  encoder/.style={draw=black,line width=.8pt,rounded corners=2pt,
    fill=figBlue!5,minimum width=1.5cm,minimum height=.76cm},
  sum/.style={circle,draw=black,line width=.85pt,fill=white,
    minimum size=.66cm,inner sep=0pt},
  every node/.style={text=black}
]
  \path[use as bounding box] (-.05,-.42) rectangle (17.7,5.30);
  \filldraw[fill=figGray!5,draw=figGray!35,rounded corners=4pt,line width=.5pt]
    (5.28,1.78) rectangle (9.86,3.88);
  % \node[text=figGray,align=center,text width=3.15cm] at (7.03,4.47)
  %   {Additive multiple-access\\channel};
  \node[text=figGray] at (7.57,1.48) {one synchronized complex use};
  % \node[text=figGray,font=\scriptsize] at (7.57,2.09) {unit effective channel gains};

  \foreach \name/\yy/\idx in {one/4.48/1,two/3.12/2,last/.72/K}{
    \node[anchor=west] (input\name) at (0,\yy) {$s_{\idx}$};
    \node[encoder] (enc\name) at (2.28,\yy) {$\enc_{\idx}(\cdot)$};
    \draw[signal] (input\name.east) -- (enc\name.west);
    \draw[draw=black,line width=.85pt] (enc\name.east) -- (4.38,\yy);
    \node[anchor=south,text=black] at (3.72,\yy+.05) {$\enc_{\idx}(s_{\idx})$};
  }
  \node at (.22,1.96) {$\vdots$};
  \node[text=black] at (2.28,1.96) {$\vdots$};
  \node[align=center] at (2.28,-.1)
    {$\sup_s|\enc_k(s)|^2\le P$};
  \node[anchor=west] at (0,5.5) {$s_k\in[-1,1]$};

  \node[sum] (mac) at (6.17,2.86) {$\sum$};
  \draw[signal] (4.38,4.48) -- (mac.135);
  \draw[signal] (4.38,3.12) -- (mac.175);
  \draw[signal] (4.38,.72) -- (mac.225);
  \node[sum] (add) at (9.08,2.86) {$+$};
  \draw[signal] (mac.east) -- node[above=3pt] {$r_{\enc}(\boldsymbol s)$} (add.west);
  \node[text=figOrange] (noise) at (9.08,5.02) {$z$};
  \draw[draw=figOrange,line width=.9pt,-{Latex[length=2.1mm]}]
    (noise.south) -- (add.north);

  \node[draw=black,line width=.85pt,rounded corners=3pt,fill=figTeal!5,
    minimum width=3.46cm,minimum height=1.05cm,align=center]
    (decoder) at (12.96,2.86)
    {Decoder $\dec$\\[2pt]
     $\widehat f=\dec(y)$};
  \draw[signal] (add.east) -- node[above=3pt] {$y$} (decoder.west);
  \draw[draw=black,line width=.85pt,-{Latex[length=2.1mm]}]
    (decoder.east) -- (16.05,2.86);
  \node[anchor=west] at (16.10,2.86) {$\widehat f$};
  \node[align=center] at (12.98,1.50)
    {$\displaystyle y=\sum_{k=1}^{K}\enc_k(s_k)+z$};
  % \node[align=center,text=figGray] at (12.98,.42)
  %   {desired function $f_{\phi,\psi}=m_2-\psi(m_1)$\\[3pt]
  %    $m_1=K^{-1}\sum_k s_k,\qquad m_2=K^{-1}\sum_k\phi(s_k)$};
\end{tikzpicture}
% END figures/tikz/system_model
}
 \caption{{\color{black} One-use function computation over a complex multiple-access channel. Each transmitter applies a continuous encoder under the pointwise peak constraint. The receiver estimates $f:\mathbb I^K\to\R$ from the noisy superposition. Channel compensation gives unit effective gains; the payload count excludes calibration and synchronization.}}
 \label{fig:system}
\end{figure*}
%-------------------------

In this paper, we address this challenge for continuous inputs by designing continuous encoders $\mathscr E_k:[-1,1]\to\C$ subject to pointwise peak-power constraints. We focus on functions, such as empirical variance, that depend on two additive statistics of locally transformed inputs and can be evaluated in a single complex channel use. We first establish conditions under which exact noiseless computation forces every continuous encoder collection to share a common affine structure that preserves both statistics, reducing the infinite-dimensional functional design problem to optimizing a finite set of complex coefficients. To select these coefficients for reliable recovery in noise, we extend the distance-based design principle of ChannelComp~\cite{channelcomp2024} to continuous domains. Specifically, we introduce a robustness score measuring noiseless received-signal separation relative to differences in the desired function value, derive an exact convex program to maximize this score, and construct decoders that achieve the resulting uniform bound on noise amplification.

\subsection{Contributions}

We connect the structural rigidity required for identifiable computation to convex robust design and sharp asymptotic estimation limits:

\begin{itemize}
\item \textbf{Encoder characterization:} For $K\ge4$ and the class of receiver postprocessing functions specified in Theorem~\ref{thm:rigidity}, we prove that exact noiseless computation forces every continuous encoder collection to take the form $\enc_k(s)=c_k+u\phi_{1,k}(s)+v\phi_{2,k}(s)$, with common coefficients $u,v\in\C$ linearly independent over $\R$. Remarkably, this characterization assumes only encoder continuity, requires neither continuity nor differentiability of the decoder, and accommodates noninvertible local preprocessing functions $\phi_{1,k},\phi_{2,k}$ with affine subintervals.

\item \textbf{Robust encoder design and decoding:}
In Theorem~\ref{thm:sdp}, we formulate robustness maximization over the characterized encoder coefficients as an exact convex semi-infinite program, eliminating nonconvex rank constraints without a relaxation gap. Under the conditions of Theorem~\ref{thm:rigidity}, this program achieves global optimality over all continuous peak-feasible encoders. In Theorem~\ref{thm:decoder}, we use McShane's Lipschitz extension~\cite{mcshane1934} to construct a decoder attaining the smallest uniform noise-amplification constant. We further derive closed-form robustness-optimal quadratic and quartic encoders and numerically certify a near-optimal exponential design.

\item \textbf{Gaussian estimation limits:}
Within the characterized encoder family, we establish sharp first-order mean squared error (MSE) limits as the Gaussian noise power vanishes. Under mild regularity conditions, Proposition~\ref{prop:general-small-noise} characterizes the minimum leading average-MSE coefficient for any fixed encoder and shows that the decoder of Theorem~\ref{thm:decoder} attains it universally across regular moment distributions. Theorem~\ref{thm:identifiable-minimax} establishes the corresponding minimax limit for worst-case MSE across all input vectors and optimizes this limit over the identifiable encoder family.

\end{itemize}

Whereas our robustness criterion controls estimation error uniformly for every input and noise realization, average Gaussian MSE weights local sensitivity by the input distribution, revealing a fundamental design tradeoff between worst-case resilience and prior-aware performance. Finally, we accompany our theoretical results with a Lean~4 formalization, with the complete artifact and verification instructions available at our \href{https://github.com/SaeedRazavikia/Rigidity-and-Limits-of-Continuous-Encoders}{GitHub repository}.

\subsection{Related Work}

Over-the-air computation naturally supports sums and linear combinations, including via lattice-based compute-and-forward~\cite{nazer2007,nazer2011}. Through local preprocessing and receiver postprocessing, analog schemes extend this additive principle to nonlinear nomographic functions such as norms and geometric means of positive inputs~\cite{golden2013harnessing,goldenbaum2015}. Because a single real aggregate with continuous processing cannot represent every continuous multivariate function, broader function classes typically require either multiple aggregates or approximation. For instance, general continuous functions on compact domains can be decomposed into sums of nomographic terms computed over successive aggregation stages~\cite{goldenbaum2015}, approximated via monotone postprocessing~\cite{limmer2015}, or learned using vector representations across multiple subcarriers~\cite{bouzinis2024}.

When inputs are drawn from finite alphabets, alternative digital and algebraic strategies become available. Arbitrary symmetric functions—such as variance, extrema, and medians—can be computed by recovering the input histogram via lattice-based block coding~\cite{jeon2014} or multidimensional symbol constructions~\cite{sahin2026generic}, while type-threshold functions and majority votes can be evaluated through multi-round interactive schemes~\cite{wang2015} or one-bit digital aggregation~\cite{zhu2021onebit}. More generally, ChannelComp designs complex constellations whose received sums distinguish different target outputs, enabling exact noiseless computation of arbitrary functions over finite alphabets~\cite{channelcomp2024}, whereas SumComp leverages algebraic constructions over standard digital modulations for sum computation~\cite{sumcomp2025,towards2025}.

In contrast, we study continuous, real-valued inputs and target functions of two additive statistics, such as empirical variance, over a single complex channel use subject to pointwise peak-power constraints. We establish when exact noiseless recovery forces every continuous encoder collection to preserve both statistics, reducing the infinite-dimensional design problem to a common finite-dimensional family. Adapting the distance-based design principle of ChannelComp~\cite{channelcomp2024} to continuous domains, we then maximize received-signal separation relative to variations in the desired function value to control estimation error in noise.

\subsection{Paper Organization}

We specify the model and desired functions in Section~\ref{sec:model}. In Section~\ref{sec:design}, we characterize the encoders and develop a robust design framework. Next, we establish decoding guarantees in Section~\ref{sec:decoding}, including Gaussian estimation limits in Section~\ref{sec:pareto}. Then, we present explicit examples and numerical evaluations in Section~\ref{sec:examples}, and conclude in Section~\ref{sec:boundaries}, with technical proofs deferred to the appendices.

\subsection{Notation}\label{sec:notation}

We use bold symbols for real vectors and represent a complex symbol $a$ by its in-phase and quadrature components, $\boldsymbol a=(\Re a,\Im a)^\T$. Let $\mathbb I=[-1,1]$ and $\mathbb I^\circ=(-1,1)$. We write $C(\mathbb I)$ for continuous real-valued functions, specifying other codomains explicitly. For $C^1(\mathbb I)$, the derivative must extend continuously and finitely to both endpoints, where derivatives are one-sided.

We use $|\cdot|$ for complex magnitude or the Euclidean norm of the corresponding real vector, $\|\cdot\|_2$ for the Euclidean vector or spectral matrix norm, and $\|\cdot\|_{\mathrm F}$ for the Frobenius norm. We denote the real symmetric $d\times d$ matrices by $\mathbb S^d$ and their positive semidefinite subset by $\mathbb S^d_+$.

\section{System Model}\label{sec:model}

We consider $K$ transmitters communicating with one receiver. Transmitter (node) $k$ observes a real input $s_k\in\mathbb I$, and the receiver wishes to compute a real-valued function of the input vector $\boldsymbol s=(s_1,\ldots,s_K)^\T\in\mathbb I^K$. Transmitter $k$ uses a continuous encoder $\enc_k:\mathbb I\to\C$, applied only to its own input, and sends one complex symbol $x_k:=\enc_k(s_k)$. We write $\enc=(\enc_1,\ldots,\enc_K)$ for the encoder collection and use $\enc(s)$ for a shared encoder when all transmitters use the same map.  We impose the pointwise peak-power constraint
\begin{equation}
  \sup_{s\in\mathbb I}|\enc_k(s)|^2\le P,\qquad k=1,\ldots,K,\quad P>0.
 \label{eq:peak-model}
\end{equation}
Thus, the power upper bound $P$ holds for every possible local observation $s_k$, independently of any input distribution.

We assume synchronized, phase-aligned transmissions with unit effective channel gains. When channel compensation is used to obtain these gains, it must be compatible with the peak constraint in \eqref{eq:peak-model}. Under this effective-channel model, the complex received sample is given by
\begin{equation}
 y=\sum\nolimits_{k=1}^Kx_k+z=r_\enc(\boldsymbol s)+z,
 \label{eq:mac-model}
\end{equation}
where $r_\enc(\boldsymbol s):=\sum_{k=1}^K\enc_k(s_k)$ denotes the noiseless received aggregate and $z$ denotes additive noise.  In the Gaussian setting, we model $z$ as circularly symmetric complex Gaussian noise, independent of the inputs, with zero mean and variance $\sigma_z^2$, i.e., $z\sim\mathcal{CN}(0,\sigma_z^2)$. Its real and imaginary components are independent, each with variance $\sigma_z^2/2$. We define the per-transmitter peak signal-to-noise ratio (SNR) as $P/\sigma_z^2$.

At the receiver, a decoder $\dec:\C\to\R$ maps the received signal to an estimate of the desired function,
$$\hat f=\dec(y).$$
The decoder $\dec$ uses knowledge of the encoders $\enc_k$ and the desired function $f$ to process the in-phase and quadrature components of $y$.

\subsection{Desired Functions}

We specify the computation task through two continuous local preprocessing functions $\phi_{1,k},\phi_{2,k}:\mathbb I\to\R$ at node $k$ and two continuous receiver postprocessing functions $\psi_1,\psi_2$. The local preprocessing functions $\phi_{j,k}$ define two aggregate statistics by averaging their outputs across nodes, and the receiver postprocessing functions $\psi_1$ and $\psi_2$ specify the desired function of these statistics. We design the encoders $\mathscr E_k$ and the decoder to compute this function over the channel. The domain of $\psi_j$ contains every attainable value of the corresponding aggregate statistic, namely the interval
$$
\mathbb J_j=\left[
\frac1K\sum_{k=1}^K\min_{s\in\mathbb I}\phi_{j,k}(s),\,
\frac1K\sum_{k=1}^K\max_{s\in\mathbb I}\phi_{j,k}(s)
\right],\qquad j=1,2.
$$
Continuity ensures that each local preprocessing function $\phi_{j,k}$ takes every value between its minimum and maximum. We study desired functions of the form
\begin{equation}
 f(\boldsymbol s)=\psi_1(m_1)+\psi_2(m_2),
 \label{eq:general-model}
\end{equation}
where $m_j=\frac1K\sum_{k=1}^K\phi_{j,k}(s_k)$ for $j=1,2$. We refer to these two aggregate statistics as \emph{moments}, although the local preprocessing functions need not be powers of the inputs. We write $\boldsymbol m=(m_1,m_2)^\T$ and use $f(\boldsymbol m)=\psi_1(m_1)+\psi_2(m_2)$ when expressing the desired function in moment coordinates. The attainable moment set is
\begin{equation}
 \mathcal M_K=\bigg\{(m_1,m_2)^\T: m_j=\frac1K\sum_{k=1}^K\phi_{j,k}(s_k),~ j=1,2,~ \boldsymbol s\in\mathbb I^K\bigg\}.
 \label{eq:general-moment-set}
\end{equation}
The set $\mathcal M_K$ contains exactly the moment vectors that some input vector can produce. It need not be convex or contain a two-dimensional region: the two statistics may constrain one another. We impose the robustness constraints only on pairs of moment vectors $\boldsymbol m,\boldsymbol m'\in\mathcal M_K$, each generated by an admissible input vector in $\mathbb I^K$. Note that we consider nonconstant desired functions, since a constant function requires no communication.

\begin{remark}
Scalar nomographic computation represents the desired function as
$$
f(\boldsymbol s)=\Phi\!\left(\sum\nolimits_{k=1}^K h_k(s_k)\right),
$$
where $h_k:\mathbb I\to\R$ are local preprocessing functions and $\Phi$ is a receiver postprocessing function~\cite{goldenbaum2015}. Our complex aggregate $r_\enc$ can retain both statistics $m_1,m_2$ in \eqref{eq:general-model}. For empirical variance with $K\ge4$, Theorem~\ref{thm:rigidity} shows that both real dimensions are necessary: no collection of continuous real-valued encoders permits exact computation, even when the encoders differ across nodes.
\end{remark}

\subsection{Identifiability and the Robustness Criterion}

For a continuous, nonconstant desired function $f:\mathbb I^K\to\R$, an encoder collection $\{\mathscr{E}_k\}_{k=1}^K$ achieves \emph{identifiable computation} if and only if 
\begin{equation}
  r_\enc(\boldsymbol s) = r_\enc(\boldsymbol s') 
  \ \Longrightarrow\ 
  f(\boldsymbol s) = f(\boldsymbol s'), 
  \qquad \forall\,\boldsymbol s, \boldsymbol s' \in \mathbb I^K.
  \label{eq:exactness}
\end{equation}
Thus, inputs with the same noiseless aggregate must have the same
desired function value. This condition is independent of the input distribution. Let $\mathcal R_\enc:=r_\enc(\mathbb I^K)$ denote the set of noiseless received aggregates. Identifiability ensures that the noiseless decoder $\dec_0:\mathcal R_\enc\to\R$ is well defined by $\dec_0(r_\enc(\boldsymbol s))=f(\boldsymbol s)$.

The identifiability condition in \eqref{eq:exactness} ensures that inputs with different desired function values produce distinct noiseless aggregates. However, exact noiseless recovery alone does not guarantee robustness to channel noise, as noted in \cite{channelcomp2024}. Robust design must also relate the separation between $r_\enc(\boldsymbol s)$ and $r_\enc(\boldsymbol s')$ to the difference between the desired function values. We express this design principle by requiring the received separation to be
at least proportional to the difference in desired function values.
Specifically, we seek a finite constant $L$ such that
$$
|f(\boldsymbol s)-f(\boldsymbol s')|
\le L\,|r_\enc(\boldsymbol s)-r_\enc(\boldsymbol s')|
\qquad
\text{for all }\boldsymbol s,\boldsymbol s'\in\mathbb I^K.
$$
A smaller admissible $L$ guarantees greater received separation for a given
difference in desired function values and limits noise amplification during
decoding. We therefore seek encoders that minimize the smallest admissible
$L$ under the peak-power constraint~\eqref{eq:peak-model}. In this regard, we define the \emph{robustness score}
\begin{equation}
 \Lambda(\enc;f)
 =\inf_{\substack{\boldsymbol s,\boldsymbol s'\in\mathbb I^K\\
                  f(\boldsymbol s)\ne f(\boldsymbol s')}}
 \frac{|r_\enc(\boldsymbol s)-r_\enc(\boldsymbol s')|^2}
 {|f(\boldsymbol s)-f(\boldsymbol s')|^2}.
 \label{eq:score}
\end{equation}
The condition $\Lambda(\enc;f)>0$ guarantees identifiability in \eqref{eq:exactness} and is equivalent to the existence of a finite bound $L$. Proposition~\ref{prop:amplification} shows that
$\Lambda(\enc;f)^{-1/2}$ is the smallest uniform noise-amplification
constant achievable by a decoder that is exact on every noiseless received aggregate.

For $P>0$, we define the class of continuous peak-feasible encoder
collections by
$$
\mathcal E_P:=\left\{
(\enc_1,\ldots,\enc_K)\in C(\mathbb I;\C)^K:
\sup_{s\in\mathbb I}|\enc_k(s)|^2\le P,
\ k=1,\ldots,K\right\}.
$$
The optimal robustness score over this class is given by
$$
\Lambda^\star(P;f):=
\sup_{\enc\in\mathcal E_P}\Lambda(\enc;f).
$$
We next determine when identifiable computation forces encoders
in $\mathcal E_P$ to preserve both prescribed moments $m_1,m_2$.

\section{Encoder Design}\label{sec:design}

We first characterize the continuous encoders $\enc_k$ that permit identifiable computation of $f$, then optimize their coefficients under the peak-power constraint~\eqref{eq:peak-model}.

\subsection{Encoder Characterization}\label{sec:rigidity}

We now characterize when identifiable computation forces the encoders $\enc_k$
to preserve both prescribed moments $m_1,m_2$. We require each receiver postprocessing function $\psi_j$
to be either affine throughout its domain or nonaffine on every open subinterval. This excludes a receiver postprocessing function $\psi_j$ that is affine on an interval of positive length but is not affine throughout its domain. Analytic functions and continuously differentiable strictly convex or strictly concave functions satisfy this condition.

\begin{theorem}
\label{thm:rigidity}\label{thm:general-rigidity}
Let $K\ge4$. For every node $k$, let
$\phi_{1,k},\phi_{2,k}\in C(\mathbb I)\cap C^1(\mathbb I^\circ)$ and suppose
that $1,\phi_{1,k},\phi_{2,k}$ are linearly independent. For $j=1,2$,
let $\psi_j\in C(\mathbb J_j)\cap C^1(\mathbb J_j^\circ)$ and
suppose that $\psi_j$ is either affine on $\mathbb J_j$ or
nonaffine on every nonempty open subinterval of
$\mathbb J_j^\circ$. If both receiver postprocessing functions $\psi_1$ and $\psi_2$ are nonconstant
and at least one is nonaffine, a collection of continuous encoders
permits identifiable computation if and only if
\begin{equation}
 \mathscr E_k(s)=c_k+u\phi_{1,k}(s)+v\phi_{2,k}(s),\qquad k=1,\ldots,K,
 \label{eq:affine}
\end{equation}
where $c_1,\ldots,c_K,u,v\in\C$ and $u,v$ are linearly independent over $\R$.
\end{theorem}
\begin{proof}
   See Appendix~\ref{app:rigidity-proof}.
\end{proof}

Under \eqref{eq:affine}, the noiseless received aggregate satisfies
$r_\enc(\boldsymbol s)
=\sum_{k=1}^K c_k+K(um_1+vm_2).$
The real linear independence of $u$ and $v$ means that their in-phase and quadrature coordinate vectors form an invertible $2\times2$ matrix. The receiver can therefore recover $m_1$ and $m_2$ from the noiseless complex aggregate. The necessity statement is stronger than the original computation requirement.  Although the receiver needs only the scalar $\psi_1(m_1)+\psi_2(m_2)$, every continuous encoder collection permitting its exact recovery must produce a noiseless aggregate from which both $m_1$ and $m_2$ can be recovered. We do not assume differentiable encoders, impose regularity on the decoder, or require either local preprocessing function to be invertible.

The linear independence of $1,\phi_{1,k},\phi_{2,k}$ ensures that the two local statistics $\phi_{1,k}(s_k)$ and $\phi_{2,k}(s_k)$ are nonconstant and that neither can be expressed as an affine function of the other for all $s_k\in\mathbb I$. In Appendix~\ref{app:rigidity-proof}, we show how this assumption and exact function recovery force the common encoder coefficients in \eqref{eq:affine}.

\begin{remark}\label{rem:rigidity-smoothness}
Smoothness of $\psi_1,\psi_2$ alone does not guarantee the encoder
characterization in Theorem~\ref{thm:rigidity}. Affine regions can
permit exact computation without preserving both moments $(m_1,m_2)$.
Appendix~\ref{app:rigidity-proof} demonstrates this with
$C^\infty$ receiver postprocessing functions and a peak-feasible encoder outside
the family \eqref{eq:affine}. For receiver postprocessing functions $\psi_1,\psi_2$ with both
affine and nonaffine regions, our optimization results remain valid
within \eqref{eq:affine}, but optimality over all continuous
encoders requires a separate necessity argument.
\end{remark}

\subsection{Optimal Design}\label{sec:synthesis}

We optimize encoders of the form \eqref{eq:affine}. For a chosen
coefficient representation, we define
\begin{subequations}
\begin{align}
c&:=\frac1K\sum_{k=1}^Kc_k,
\quad \mathsf U:=
\begin{pmatrix}
 \Re (u) & \Re (v)\\
 \Im (u) & \Im (v)
\end{pmatrix},\\
G&:=\mathsf U^\T\mathsf U,
\quad H_k:=[\boldsymbol c_k\ \mathsf U]^\T
     [\boldsymbol c_k\ \mathsf U],
\end{align}
\end{subequations}
where each $H_k$ is the \emph{Gram matrix} of the three vectors
$\boldsymbol c_k,\boldsymbol u,\boldsymbol v\in\R^2$, so
$H_k\succeq0$ and $\operatorname{rank}(H_k)\le2$.
The rank constraint reflects the two real dimensions available in one complex channel use. Imposing this constraint makes optimization over the Gram matrices nonconvex.

For coefficient design, we first allow arbitrary complex coefficients in the algebraic form \eqref{eq:affine}, including linearly dependent $u,v$. We impose invertibility separately and define
$$
\begin{aligned}
\mathcal B_P^{\mathrm{all}}
&:=\{\enc\in\mathcal E_P:
\enc\text{ admits }\eqref{eq:affine}\},\\
\mathcal B_P
&:=\{\enc\in\mathcal B_P^{\mathrm{all}}:
\enc\text{ admits }\eqref{eq:affine}\text{ with }\det\mathsf U\ne0\}.
\end{aligned}
$$
The coefficient representation must hold for every node and every $s\in\mathbb I$.
In each coefficient design, $G$ and $H_k$ refer
to the selected representation. Every member of $\mathcal B_P$
preserves both moments. We also retain singular representations in
$\mathcal B_P^{\mathrm{all}}$ to include functions that can be
computed from one real aggregate, such as desired functions
depending on only one prescribed moment.
The common block $G$ determines noiseless received-signal separation, while $H_k$ also includes the offset needed to satisfy the peak-power
constraint~\eqref{eq:peak-model} at node $k$.

\begin{theorem}
\label{thm:sdp}\label{thm:general-secant}
Let $K\ge2$ and $P>0$. Assume that all local preprocessing functions $\phi_{j,k}$ and both receiver postprocessing functions
$\psi_1,\psi_2$ are continuous on their prescribed domains and that
$f$ is nonconstant on $\mathcal M_K$. Then
$$
 \max_{\enc\in{ \mathcal B_P^{\mathrm{all}}}}\Lambda(\enc;f)
 =K^2\lambda_\star,
$$
where $\lambda_\star$ is the attained optimum of the convex program
\begin{subequations}
\label{eq:general-secant-program}
\begin{align}
\nonumber
 \underset{\lambda\ge0,\ H_1,\ldots,H_K\in\mathbb S^3}
 {\operatorname{maximize}}
 &\quad \lambda\\ \label{eq:general-secant-program(a)}
 \text{subject to}\quad
 &H_k\succeq0,\quad (H_k)_{2:3,2:3}=G,~ \forall\,k
       \\ \label{eq:general-secant-program(b)}
 &(1,\phi_{1,k}(s),\phi_{2,k}(s))H_k
       (1,\phi_{1,k}(s),\phi_{2,k}(s))^\T\le P
       &&(s\in\mathbb I,\forall\,k),\\ \label{eq:general-secant-program(c)}
 &(\boldsymbol m-\boldsymbol m')^\T
   G(\boldsymbol m-\boldsymbol m')
 \ge\lambda |f(\boldsymbol m)-f(\boldsymbol m')|^2
       &&(\boldsymbol m,\boldsymbol m'\in\mathcal M_K).
\end{align}
\end{subequations}
Here, $G=(H_1)_{2:3,2:3}$ is the common lower-right block.
Moreover, $\lambda_\star>0$ if and only if $f$ is Lipschitz
on $\mathcal M_K$. 
Under the assumptions of
Theorem~\ref{thm:rigidity}, we also have
$\Lambda^\star(P;f)=K^2\lambda_\star$, so the same program
optimizes robustness over the full class $\mathcal E_P$.
\end{theorem}
\begin{proof}
See Appendix~\ref{app:general-robust}.
\end{proof}

Theorem~\ref{thm:sdp} eliminates the nonconvex rank constraints
$\operatorname{rank}(H_k)\le2$ without a relaxation gap. The resulting
program has $3K+4$ real decision variables and does not require
convexity of $\mathcal M_K$, linear independence of the local
functions $\phi_{1,k},\phi_{2,k}$, or differentiability of the receiver postprocessing functions $\psi_1,\psi_2$.
The program is semi-infinite because it has finitely many design variables but imposes power and received-separation inequalities over continuous sets of inputs and pairs of moment vectors. We reduce these constraints to finitely many conditions for the polynomial examples in Section~\ref{sec:quadratic-quartic}; such a reduction is not guaranteed in general.

To reconstruct the optimal encoders $\enc_k^\star$, we partition each optimal Gram matrix as
$H_k^\star=\left(
\begin{smallmatrix}
h_k^\star & (\boldsymbol{\ell}_k^\star)^\mathsf{T} \\
\boldsymbol{\ell}_k^\star & G^\star
\end{smallmatrix}
\right)
$. 
We choose $\mathsf{U}^\star=(G^\star)^{1/2}$ and compute the offsets as
$\boldsymbol{c}_k^\star=(G^\star)^{\dagger/2}\boldsymbol{\ell}_k^\star$,
where $\dagger$ denotes the Moore--Penrose inverse. Here $(G^\star)^{\dagger/2}$ is the positive semidefinite square root of $(G^\star)^\dagger$. This reconstruction retains the same squared received distance for every pair of moment vectors and cannot increase the transmitted power at any local input. Appendix~\ref{app:general-robust} establishes these two properties directly. When
$G^\star\succ0$, we may instead use any factor satisfying
$(\mathsf{U}^\star)^\mathsf{T}\mathsf{U}^\star=G^\star$ and set
$\boldsymbol{c}_k^\star=(\mathsf{U}^\star)^{-\mathsf{T}}
\boldsymbol{\ell}_k^\star$. The columns of $\mathsf{U}^\star$ and
the vectors $\boldsymbol{c}_k^\star$ determine the complex coefficients
$u^\star$, $v^\star$, and $c_k^\star$, respectively.

The final constraints in \eqref{eq:general-secant-program(c)} compare function values and received signals at every pair of attainable moment vectors. For $\lambda>0$, these pairwise, or secant, constraints limit noise amplification for all inputs. A local gradient bound cannot generally replace the pairwise constraints in \eqref{eq:general-secant-program(c)}. To relate the pairwise constraints to local gradient bounds, fix an encoder collection of the form \eqref{eq:affine} with $G\succ0$, and suppose that $f$ is continuously differentiable on a neighborhood of a convex set $\mathcal C\supseteq\mathcal M_K$. With
$\boldsymbol w(\boldsymbol m):=
(\psi_1'(m_1),\psi_2'(m_2))^\T$, we obtain
\begin{equation}
\sup_{\boldsymbol m\in\operatorname{int}\mathcal M_K}
\boldsymbol w(\boldsymbol m)^\mathsf T G^{-1}\boldsymbol w(\boldsymbol m)
\le \frac{K^2}{\Lambda(\enc;f)}\le \sup_{\boldsymbol m\in\mathcal C}
\boldsymbol w(\boldsymbol m)^\mathsf T G^{-1}\boldsymbol w(\boldsymbol m) =:B_{\mathcal C}(G). 
\label{eq:general-gradient-sandwich}
\end{equation}
We interpret the leftmost supremum as zero when
$\operatorname{int}\mathcal M_K=\varnothing$. If $\mathcal M_K$ is
compact, convex, and has nonempty interior, choosing
$\mathcal C=\mathcal M_K$ makes both inequalities equalities.
If the moment set $\mathcal M_K$ has no two-dimensional interior, a gradient bound that includes all directions in $\R^2$ may include perturbations that no input vector can produce and need not be sharp. Appendix~\ref{app:gradient-bounds} proves
\eqref{eq:general-gradient-sandwich} and gives a counterexample showing that using only the gradients at $m_1=\pm1$ can overestimate the robustness score $\Lambda(\enc;f)$.

\begin{remark}[Numerical certification]\label{remark:certification}
In Section~\ref{sec:quadratic-quartic}, we derive analytical
quadratic and quartic designs satisfying all constraints in
\eqref{eq:general-secant-program}. For the function class in Lemma~\ref{prop:score}, including the exponential example, the pairwise received-separation constraints reduce to two matrix inequalities. We verify these inequalities for the exponential design using interval arithmetic. We also use interpolation bounds to enforce the
peak constraints between grid points. These checks certify
feasibility, while primal and dual certificates provide lower
and upper bounds on $\lambda_\star$; see Appendix~\ref{app:example-proofs} (in particular, Appendices~\ref{sec:certified} and~\ref{app:certificate-witness-bounds}).
\end{remark}

\section{Decoding}\label{sec:decoding}

Theorem~\ref{thm:rigidity} characterizes encoders that enable exact recovery in the absence of noise. In the presence of additive noise, however, the received signal $y$ may fall outside the set of noiseless values $\mathcal R_\enc$. We therefore construct a decoder $\dec:\C\to\R$ that is exact without noise and minimizes uniform noise amplification. 

For a fixed encoder collection $\enc$, a decoder $\dec:\C\to\R$ has uniform noise-amplification constant $L\ge0$ if
\begin{equation}
|\dec(r_\enc(\boldsymbol s)+z)-f(\boldsymbol s)|
\le L|z|,
\qquad
\forall\,\boldsymbol s\in\mathbb I^K,\ \forall\,z\in\C.
\label{eq:uniform}
\end{equation}
This deterministic guarantee holds for every noise realization $z$ and requires no assumption on its distribution. Whenever $|z|\le\varepsilon$, the absolute decoding error is at most $L\varepsilon$. Setting $z=0$ also enforces exact recovery under noiseless conditions. We note that the uniform bound in \eqref{eq:uniform} is closely related to deterministic noise-robust recovery guarantees for linear inverse problems~\cite{bourrier2014fundamental}; here, we measure the error in the desired function value.

For the prescribed desired function $f$, we define
$$
L_\enc:=\inf\left\{L\ge0:
\text{there exists }\dec:\C\to\R
\text{ satisfying }\eqref{eq:uniform}\right\},
$$
with $\inf\varnothing:=+\infty$. Thus, $L$ denotes an admissible
uniform bound and $L_\enc$ denotes the best achievable constant
for the fixed encoder collection $\enc$. We suppress its dependence on the
prescribed function $f$. Proposition~\ref{prop:amplification}
relates $L_\enc$ to the robustness score $\Lambda(\enc;f)$ in \eqref{eq:score}.

%----------------------------------
\begin{figure}[t]
 \centering\paperfigure{%
% BEGIN figures/tikz/decoder_geometry
% Native TikZ: exact K=4 quadratic moment region in normalized received coordinates.
% a=sqrt(b(2-b)), b=2/sqrt(3), P-normalized coordinates xi_1=Re(y)/(K sqrt(P)),
% xi_2=Im(y)/(K sqrt(P)); all coordinates use the same scale (Euclidean geometry).
\begin{tikzpicture}[
  x=1.90cm,y=1.90cm,>=Latex,font=\footnotesize,
  axis/.style={draw=figGray!70,line width=.55pt,-{Latex[length=1.6mm]}},
  every node/.style={text=black}
]
  \pgfmathsetmacro{\bb}{2/sqrt(3)}
  \pgfmathsetmacro{\aa}{sqrt(\bb*(2-\bb))}
  \path[use as bounding box] (-1.52,-1.36) rectangle (2.19,1.31);
  % The strip is unbounded vertically; dashed edges continue beyond the view.
  \fill[figBlue!5] (-\aa,-1.14) rectangle (\aa,1.12);
  \draw[figBlue!45,dashed,line width=.7pt] (-\aa,-1.14) -- (-\aa,1.12);
  \draw[figBlue!45,dashed,line width=.7pt] (\aa,-1.14) -- (\aa,1.12);
  % Exact attainable region: lower boundary m_2=m_1^2, four upper-bound arcs.
  \path[fill=figTeal!22,draw=figTeal,line width=.75pt]
    plot[domain=-1:1,samples=81,variable=\m]
      ({\aa*\m},{\bb*\m*\m-1})
    -- plot[domain=1:.5,samples=31,variable=\m]
      ({\aa*\m},{\bb*(3+(4*\m-3)^2)/4-1})
    -- plot[domain=.5:0,samples=31,variable=\m]
      ({\aa*\m},{\bb*(3+(4*\m-1)^2)/4-1})
    -- plot[domain=0:-.5,samples=31,variable=\m]
      ({\aa*\m},{\bb*(3+(4*\m+1)^2)/4-1})
    -- plot[domain=-.5:-1,samples=31,variable=\m]
      ({\aa*\m},{\bb*(3+(4*\m+3)^2)/4-1}) -- cycle;
  \draw[axis] (-1.26,-1.16) -- (1.83,-1.16);
  \draw[axis] (-1.26,-1.16) -- (-1.26,1.13);
  \node[anchor=west] at (-1.25,1.24) {$\xi_2=\Im(y)/(K\sqrt P)$};
  \node[anchor=north east,font=\scriptsize] at (1.85,-1.23)
    {$\xi_1=\Re(y)/(K\sqrt P)$};
  \foreach \zz in {-1,0}{
    \draw[figGray!70,line width=.45pt] (-1.29,\zz) -- (-1.23,\zz);
    \node[anchor=east,font=\scriptsize] at (-1.29,\zz) {$\zz$};
  }
  \foreach \xx/\txt in {-\aa/-a,0/0,\aa/a}{
    \draw[figGray!70,line width=.45pt] (\xx,-1.13) -- (\xx,-1.19);
    \node[anchor=north,font=\scriptsize] at (\xx,-1.18) {$\txt$};
  }
  \node[text=figBlue,align=center,font=\scriptsize] at (-.05,.92)
    {containing strip\\$|\xi_1|\le a$};
  \node[text=figTeal!80!black,font=\scriptsize,align=center] at (-.15,-.34)
    {attainable region\\$\mathcal R_{\enc}/(K\sqrt P)$};
  % A valid K=4 moment pair is (m_1,m_2)=(.65,.61), so its received point is feasible.
  \coordinate (r) at ({.65*\aa},{.61*\bb-1});
  \coordinate (y) at (1.49,.61);
  \coordinate (p) at (\aa,.61);
  \draw[draw=figOrange,line width=.85pt,-{Latex[length=1.8mm]}]
    (r) -- node[pos=.30,right=3pt,text=figOrange,font=\scriptsize]
      {$z/(K\sqrt P)$} (y);
  \draw[draw=figBlue,line width=1pt,-{Latex[length=1.8mm]}] (y) -- (p);
  \fill[figTeal] (r) circle[radius=.025];
  \fill[figOrange] (y) circle[radius=.025];
  \fill[figBlue] (p) circle[radius=.025];
  \draw[figBlue,line width=.6pt] (\aa,.68) -- ({\aa+.07},.68)
    -- ({\aa+.07},.61);
  \node[anchor=north west,font=\scriptsize] at ($(r)+(.08,-.08)$)
    {$r_{\enc}(\boldsymbol s)/(K\sqrt P)$};
  \node[anchor=south west,text=figOrange,font=\scriptsize] at ($(y)+(.015,.045)$)
    {$\boldsymbol\xi$};
  \node[anchor=south east,text=figBlue,font=\scriptsize] at ($(p)+(-.025,.055)$)
    {$\Pi_{\mathcal S_r}(\boldsymbol\xi)$};
\end{tikzpicture}
% END figures/tikz/decoder_geometry
}
 \caption{Projection decoding for the optimal variance encoder with $K=4$, $b=2/\sqrt3$, and $a=\sqrt{b(2-b)}$. We plot the normalized received vector $\boldsymbol\xi=(\Re y,\Im y)^\top/(K\sqrt P)$ using equal axis scales. The decoder projects $\boldsymbol\xi$ onto the strip $\mathcal S_r=\{(\xi_1,\xi_2):|\xi_1|\le a\}$, recovers $\widehat m_1,\widehat m_2$, and evaluates $\widehat m_2-\widehat m_1^2$. The strip contains all normalized noiseless received points, but the projected point need not be attainable.}
 \label{fig:decoder-geometry}
\end{figure}
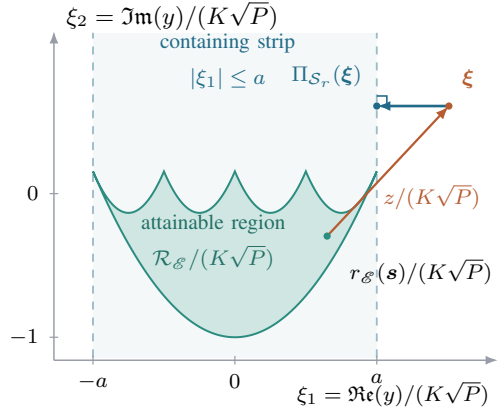
%----------------------------------

\begin{proposition}\label{prop:amplification}
Let $K\ge1$, let $\enc$ be a collection of continuous encoders, and let $f:\mathbb I^K\to\R$ be continuous and nonconstant. Then
$${ L_\enc}=\Lambda(\enc;f)^{-1/2},$$
with $0^{-1/2}=+\infty$. A continuous decoder attains the infimum whenever it is finite.
\end{proposition}
\begin{proof}
    See Appendix~\ref{app:amplification-proof}.
\end{proof}
Maximizing $\Lambda(\enc;f)$ therefore minimizes uniform noise amplification. Theorem~\ref{thm:decoder} constructs a decoder attaining $L_\enc$ for any attainable moment set $\mathcal M_K$.

\begin{theorem}
\label{thm:decoder}\label{thm:general-decoder}

Under the assumptions of Theorem~\ref{thm:sdp}, fix an encoder collection of the form \eqref{eq:affine} with $\Lambda(\enc;f)>0$.
Using $c$ and $\mathsf U$ from Section~\ref{sec:synthesis}, the decoder
\begin{equation}
\dec_\enc(y):=
\min_{\boldsymbol m\in\mathcal M_K}
\left\{f(\boldsymbol m)
+L_\enc\|\boldsymbol y-K\boldsymbol c-K\mathsf U\boldsymbol m\|_2
\right\}
\label{eq:general-extension-decoder}
\end{equation}
satisfies
$$
|\dec_\enc(r_\enc(\boldsymbol s)+z)-f(\boldsymbol s)|
\le L_\enc|z|,
\qquad
\forall\,\boldsymbol s\in\mathbb I^K,\ \forall\,z\in\C.
$$
It therefore attains the optimal uniform noise-amplification
constant $L_\enc=\Lambda(\enc;f)^{-1/2}$.

\end{theorem}
\begin{proof}
    The proof is given in Appendix~\ref{app:general-decoder-proof}.
\end{proof}

The bound in Theorem~\ref{thm:decoder} ensures exact recovery when $z=0$. For any noise
distribution satisfying $\mathbb E|z|^2=\sigma_z^2$, it also gives
$$
\sup_{\boldsymbol s\in\mathbb I^K}
\mathbb E_z\!\left[
|{ \dec_\enc}(r_\enc(\boldsymbol s)+z)-f(\boldsymbol s)|^2
\right]
\le \frac{\sigma_z^2}{\Lambda(\enc;f)}.
$$
Clipping the decoder output $\dec_\enc(y)$ to a closed interval containing
$f(\mathbb I^K)$ preserves exact recovery and both error bounds.
The minimization in \eqref{eq:general-extension-decoder} uses
$\mathcal M_K$ directly, including when it is nonconvex or has
empty interior.

When the minimization in \eqref{eq:general-extension-decoder} is difficult, we may use a closed convex set
$\mathcal C\supseteq\mathcal M_K$. Suppose that $\mathsf U$ is
invertible, $f$ is continuously differentiable on a neighborhood
of $\mathcal C$, and the gradient bound $B_{\mathcal C}(G)$ in \eqref{eq:general-gradient-sandwich} is finite.
We use the weighted norm
$\|\boldsymbol x\|_G:=(\boldsymbol x^\T G\boldsymbol x)^{1/2}$.
We first form
$\tilde{\boldsymbol m}=K^{-1}\mathsf U^{-1}
(\boldsymbol y-K\boldsymbol c)$, then project it onto $\mathcal C$:
$$
\hat{\boldsymbol m}=\Pi_{\mathcal C}^G
(\tilde{\boldsymbol m}),\qquad
\Pi_{\mathcal C}^G(\boldsymbol x):=
\operatorname*{arg\,min}_{\boldsymbol m\in\mathcal C}
(\boldsymbol m-\boldsymbol x)^\T G(\boldsymbol m-\boldsymbol x).
$$
We decode by evaluating
$\dec_{\mathcal C}(y)=f(\hat{\boldsymbol m})$.
Because $\mathcal M_K\subseteq\mathcal C$, projection leaves every noiseless moment vector unchanged. It also cannot increase the distance from the noisy estimate to the true moment vector in the $G$ norm; this norm measures the corresponding received-signal separation up to the factor $K$. Since
$\sqrt{B_{\mathcal C}(G)}$ bounds the sensitivity of $f$ in this
metric, the decoding error is at most
$\sqrt{B_{\mathcal C}(G)}\,|z|/K$. The decoder $\dec_{\mathcal C}$ attains the optimal
uniform noise-amplification constant when
$B_{\mathcal C}(G)=K^2/\Lambda(\enc;f)$.

Theorem~\ref{thm:decoder} minimizes the constant $L$ in
\eqref{eq:uniform} for a fixed encoder collection $\enc$, using a single decoder for every input vector and every noise value. With a known noise bound $|z|\le r_n$, Proposition~\ref{prop:radius-minimax} gives the minimum worst-case absolute error at that radius. 

\begin{proposition}\label{prop:radius-minimax}
For a fixed continuous encoder collection $\enc$, a continuous desired function $f$,
and a known noise bound $|z|\le r_n$ with $r_n\ge0$, the minimum worst-case
absolute error is $\omega_\enc(2r_n)/2$, where
$$
\omega_\enc(d):=
\max_{\substack{
\boldsymbol s,\boldsymbol s'\in\mathbb I^K\\
|r_\enc(\boldsymbol s)-r_\enc(\boldsymbol s')|\le d}}
|f(\boldsymbol s)-f(\boldsymbol s')|,
\qquad d\ge0.
$$
An optimal decoder outputs the midpoint between the smallest and largest
values of $f$ consistent with the observation and noise bound.
\end{proposition}
\begin{proof}
    The proof is given in Appendix~\ref{app:radius-minimax-proof}.
\end{proof}

\subsection{Gaussian Estimation Limits}\label{sec:pareto}

Theorem~\ref{thm:decoder} bounds the decoding error for every noise realization. For $z\sim\mathcal{CN}(0,\sigma_z^2)$, we use classical asymptotic estimation bounds~\cite{gill1995,tsang2020,gassiat2024} to characterize the minimum MSE as $\sigma_z^2\downarrow0$, with $P$, $K$, the desired function $f$, and any prescribed moment distribution fixed. Throughout this subsection, we assume $\psi_j\in C^1(\mathbb J_j)$ for $j=1,2$.

% Theorem~\ref{thm:decoder} bounds the decoding error for every noise realization. For $z\sim\mathcal{CN}(0,\sigma_z^2)$, we characterize the minimum MSE as
% $\sigma_z^2\downarrow0$, with $P$, $K$, the desired function $f$, and any prescribed moment distribution fixed. Throughout this subsection, we assume $\psi_j\in C^1(\mathbb J_j)$ for $j=1,2$.

For a fixed encoder collection $\enc\in\mathcal B_P$ and a Borel-measurable
decoder $\dec:\C\to\R$, we define the worst-case Gaussian MSE~\cite{jiao2015minimax} as
\begin{align*}
    R_{\max}(\enc,\dec;\sigma_z^2)
:=
\sup_{\boldsymbol m\in\mathcal M_K}
\mathbb E_z\!\left[
\left|\dec\!\left(Kc+K(um_1+vm_2)+z\right)
-f(\boldsymbol m)\right|^2
\right].
\end{align*}
The expectation is over channel noise, while the supremum covers all attainable moment vectors. We optimize over all Borel-measurable decoders, so the expected error is well defined. In particular, we do not require the decoder to be continuous, unbiased, or exact without noise. The decoder may also depend on the noise power $\sigma_z^2$. The leading MSE coefficient is the limiting ratio of MSE to noise power as $\sigma_z^2\downarrow0$. In Theorem~\ref{thm:identifiable-minimax}, we determine this coefficient both for each fixed encoder collection and after optimization over $\mathcal B_P$.

The geometric assumption $\mathcal M_K=\overline{\operatorname{int}\mathcal M_K}$ means that every attainable moment vector can be approached by attainable moment vectors that admit small perturbations in both coordinates. It allows us to extend estimation lower bounds obtained at interior moment vectors to the entire moment set.

\begin{theorem}\label{thm:identifiable-minimax}
Assume $\mathcal M_K=\overline{\operatorname{int}\mathcal M_K}$.
For every fixed encoder collection $\enc\in\mathcal B_P$,
\begin{equation}
 \lim_{\sigma_z^2\downarrow0}
 \frac{\inf_{\dec} R_{\max}(\enc,\dec;\sigma_z^2)}{\sigma_z^2}
 =\frac{1}{2K^2}\max_{\boldsymbol m\in\mathcal M_K}
 \boldsymbol w(\boldsymbol m)^\T G^{-1}\boldsymbol w(\boldsymbol m).
 \label{eq:fixed-minimax-limit}
\end{equation}
If, in addition, $1,\phi_{1,k},\phi_{2,k}$ are linearly independent
for every node $k$ and the vectors
$\{\boldsymbol w(\boldsymbol m):\boldsymbol m\in
\operatorname{int}\mathcal M_K\}$ span $\R^2$, then
\begin{equation}
 \lim_{\sigma_z^2\downarrow0}
 \frac{\inf_{\enc\in\mathcal B_P}\inf_{\dec}
 R_{\max}(\enc,\dec;\sigma_z^2)}{\sigma_z^2}=
 \min_{\enc\in\mathcal B_P}\frac{1}{2K^2}
 \max_{\boldsymbol m\in\mathcal M_K}
 \boldsymbol w(\boldsymbol m)^\T G^{-1}\boldsymbol w(\boldsymbol m).
 \label{eq:optimized-minimax-limit}
\end{equation}
The minimum is attained. Both asymptotic coefficients are achievable using encoders and decoders that are independent of the noise power.
\end{theorem}
\begin{proof}
See Appendix~\ref{app:identifiable-minimax-proof} for the decoder
construction and proof.
\end{proof}

Equation~\eqref{eq:optimized-minimax-limit} justifies optimizing the fixed-encoder limit in \eqref{eq:fixed-minimax-limit} over $\mathcal B_P$. Under the additional assumptions, a fixed encoder--decoder pair attains the optimum even when competing designs vary with noise power. The spanning condition requires two interior points of $\mathcal M_K$ at which the gradients $\boldsymbol w(\boldsymbol m)$ are nonzero and not parallel. In simple terms, the relative sensitivities of the desired function to $m_1$ and $m_2$ must differ between these points. Allowing encoder collections in $\mathcal B_P$ and their decoders to depend on noise power cannot improve this coefficient. When Theorem~\ref{thm:rigidity} applies, $\mathcal B_P$ contains all continuous peak-feasible encoder collections that permit exact noiseless computation, so this optimization also covers that entire identifiable class.

The Gaussian coefficient in \eqref{eq:fixed-minimax-limit} depends
on the local sensitivity $\boldsymbol w$, whereas $\Lambda(\enc;f)$
uses all attainable moment vectors. These criteria can therefore
favor different designs. Appendix~\ref{genstat:counterexamples}
shows that the Gaussian coefficient need not equal
$1/[2\Lambda(\enc;f)]$.

\begin{remark}
\label{rem:moment-set-condition}
The condition
$\mathcal M_K=\overline{\operatorname{int}\mathcal M_K}$
in Theorem~\ref{thm:identifiable-minimax} does not follow
from linear independence of the local preprocessing
functions. To see this, let every node use
$\phi_1(s)=s$ and $\phi_2(s)=(\max\{s,0\})^2$.
These functions are $C^1$, and $1,\phi_1,\phi_2$ are
linearly independent. If any input is nonnegative, then
$m_1\ge-(K-1)/K$. Consequently, every attainable moment vector
with $m_1<-(K-1)/K$ has $m_2=0$.
In this region, the moment set is confined to a line.
The attainable points $(t,0)$, obtained by setting all
inputs equal to $t<-(K-1)/K$, therefore lie outside
$\overline{\operatorname{int}\mathcal M_K}$.
\end{remark}

For a given moment distribution $\pi$ on $\mathcal M_K$,
we also define the average Gaussian MSE~\cite{jiao2015minimax} as
\begin{align*}
    R_\pi(\enc,\dec;\sigma_z^2)
 :=
 \mathbb E_{\pi,z}\!\left[
 \left|\dec\!\left(Kc+K(um_1+vm_2)+z\right)
 -f(\boldsymbol m)\right|^2
 \right].
\end{align*}
We summarize the average local sensitivity by $ Q_\pi:=
 \mathbb E_\pi[
 \boldsymbol w(\boldsymbol m)\boldsymbol w(\boldsymbol m)^\T].$ For $G\succ0$, we define the corresponding average-MSE coefficient
\begin{equation}
\begin{aligned}
 \mathcal J_\pi(G)
:=
 \frac{1}{2K^2}\mathbb E_\pi\!\left[
 \boldsymbol w(\boldsymbol m)^\T G^{-1}
 \boldsymbol w(\boldsymbol m)\right]=\frac{\operatorname{tr}(Q_\pi G^{-1})}{2K^2}.
\end{aligned}
 \label{eq:average-mse-coefficient}
\end{equation}
The coefficient $\mathcal J_\pi(G)$ is the prior average of the
local Cram\'er--Rao expression for estimating $f(\boldsymbol m)$,
normalized by the noise power~\cite{tsang2020}. Optimizing $\mathcal J_\pi(G)$ has the weighted inverse-trace form of A-optimal experimental design~\cite{sagnol2010,papp2011}. We call the distribution $\pi$ a \emph{regular moment distribution} if it has a strictly positive, continuously differentiable density on an open set $\Omega\subset\operatorname{int}\mathcal M_K$ and $\pi(\Omega)=1$. Thus, all probability is assigned to interior moment vectors, although samples may lie arbitrarily close to the boundary. Proposition~\ref{prop:general-small-noise} identifies $\mathcal J_\pi(G)$ as the minimum leading average-MSE coefficient for a fixed encoder collection $\enc$.

\begin{proposition}
\label{prop:general-small-noise}\label{genstat:fixed}
Fix $\enc\in\mathcal B_P$ and suppose that $\pi$ is a regular moment distribution. Then
\begin{equation}
 \lim_{\sigma_z^2\downarrow0}
 \frac{\inf_{\dec}R_\pi(\enc,\dec;\sigma_z^2)}{\sigma_z^2}
 =\mathcal J_\pi(G).
 \label{eq:fixed-bayes-limit}
\end{equation}
The decoder ${ \dec_\enc}$ of Theorem~\ref{thm:decoder} attains this coefficient.
\end{proposition}
\begin{proof}
See Appendix~\ref{genstat:information}.
\end{proof}

Thus, for a fixed encoder collection $\enc$ and any regular moment distribution $\pi$, $\dec_\enc$
attains both the optimal uniform noise-amplification constant $L_\enc$ and the
minimum leading average-MSE coefficient $\mathcal J_\pi(G)$, without depending on $\pi$.

\section{Explicit Designs and Performance Evaluation}\label{sec:examples}

We derive robustness-optimal quadratic and quartic encoders and
certify a numerical exponential design. Then, we use the polynomial
examples to compare robustness and Gaussian MSE under different
encoder choices. In this section, we consider desired functions of the form
\begin{equation}
 f_{\phi,\psi}(\boldsymbol s)
 =\frac1K\sum\nolimits_{k=1}^K\phi(s_k)
 -\psi\!\left(\frac1K\sum\nolimits_{k=1}^K s_k\right).
 \label{eq:target}
\end{equation}
Here $m_1=K^{-1}\sum_k s_k$ and
$m_2=K^{-1}\sum_k\phi(s_k)$, so
$f_{\phi,\psi}(\boldsymbol s)=m_2-\psi(m_1)$.  Within the encoder family \eqref{eq:affine}, we may restrict
attention to shared encoders because all nodes use the same
local preprocessing functions. Replacing the offsets $c_k$
by their average $c=K^{-1}\sum_k c_k$ preserves every
received aggregate and maintains peak feasibility, since
$$
|c+us+v\phi(s)|^2
=\left|\frac1K\sum_k\enc_k(s)\right|^2
\le\frac1K\sum_k|\enc_k(s)|^2
\le P.
$$
This replacement therefore preserves the robustness score
and the error of any fixed decoder. To optimize the shared encoder coefficients, we first express
the robustness score in terms of the Gram matrix $G$.
The following lemma shows that, under the stated convexity
or concavity assumptions, this score depends on $\psi'$
only through its endpoint values.

\begin{lemma}\label{prop:score}
Let $K\ge4$. Assume that
$\phi\in C(\mathbb I)\cap C^1(\mathbb I^\circ)$ and
$\psi\in C^1(\mathbb I)$ are each strictly convex or strictly
concave. For the desired function \eqref{eq:target}, every
encoder collection $\enc_k(s)=c_k+us+v\phi(s)$ with $G\succ0$
satisfies
\begin{equation}
 \Lambda(\enc;f_{\phi,\psi})
 =\frac{K^2}{\max_{j\in\{-,+\}}\boldsymbol w_j^\T G^{-1}\boldsymbol w_j},
 \label{eq:exactscore}
\end{equation}
where $\boldsymbol w_\pm=(-\psi'(\pm1),1)^\T$.
The infimum defining $\Lambda(\enc;f_{\phi,\psi})$ in \eqref{eq:score} is approached by attainable pairs whose
noiseless aggregate distance tends to zero.
\end{lemma}
\begin{proof}
    See Appendix~\ref{app:prof:score}. 
\end{proof}

\subsection{Optimal Quadratic and Quartic Encoders}\label{sec:quadratic-quartic}

In Proposition~\ref{thm:parabola}, we determine the shared encoder
coefficients that maximize $\Lambda(\enc;f_{\phi,\psi})$ for
$\phi(s)=s^2$ and $\phi(s)=s^4$.

\begin{proposition}
\label{thm:parabola}
Let $K\ge4$, $P>0$, and let $\psi\in C^1(\mathbb I)$ be even
and strictly convex or strictly concave. Write $\chi_p(s)=s^p$
and $M=|\psi'(1)|$. For \eqref{eq:target} with
$\phi=\chi_2$ and $\phi=\chi_4$, respectively, the shared encoders
\begin{equation}
 \enc_{\chi_2}^\star(s)=\sqrt P\,[a s+i(bs^2-1)], \quad 
\Lambda^\star(P;f_{\chi_2,\psi})=K^2P b^2(b-1),\quad s\in\mathbb I,
 \label{eq:parabola}
\end{equation}
and
\begin{equation}
 \enc_{\chi_4}^\star(s)=\sqrt P\,[a_\rho s+i(b_\rho s^4+c_\rho)],\quad
\Lambda^\star(P;f_{\chi_4,\psi})=\frac{K^2P}{M^2/a_\rho^2+1/b_\rho^2},\quad s\in\mathbb I,
 \label{eq:quarticscore}
\end{equation}
maximize robustness over all continuous encoder collections
satisfying \eqref{eq:peak-model}. The quadratic coefficients are $b=\frac{2\sqrt{M^2+8}+2M}{\sqrt{M^2+8}+3M}$ and $a=\sqrt{b(2-b)}$. For the quartic encoder, $\rho\in(0,1/3)$ is the unique solution of
\begin{equation}
 8\rho^2(1+\rho)(1+3\rho)=M^2(1-3\rho),
 \label{eq:quarticroot}
\end{equation}
and $d_\rho=1+2\rho+9\rho^2$, $a_\rho=\sqrt{\frac{8\rho}{d_\rho}}$, $b_\rho=\frac{2}{(1+\rho)\sqrt{d_\rho}}$, and $c_\rho=-\frac{1+2\rho+3\rho^2}{(1+\rho)\sqrt{d_\rho}}$.
\end{proposition}
\begin{proof}
See Appendix~\ref{app:quartic-proof}.
\end{proof}

For variance $V(\boldsymbol s):=m_2-m_1^2$,
$\phi(s)=\psi(s)=s^2$ gives $M=2$,
$b=2/\sqrt3$, and
$\Lambda^\star(P;V)=4K^2P/(9+6\sqrt3)$.
For the fourth-order Jensen gap, $\phi(s)=\psi(s)=s^4$ gives
$M=4$ and an optimal normalized robustness score of approximately
$0.058533507$.
Figure~\ref{fig:curves} compares the optimal encoders with the
peak-feasible baselines $\sqrt P[s+i(s^2-1)]$ and
$(2\sqrt P/\sqrt5)[s+i(s^4-1/2)]$, respectively.\footnote{For
$|s|\le1$, the normalized powers satisfy
$s^2+(s^2-1)^2=1-s^2(1-s^2)\le1$ and
$\frac45[s^2+(s^4-\frac12)^2]
\le\frac45(1+\frac14)=1$.}
Equation~\eqref{eq:exactscore} gives their robustness scores
as $K^2P/5$ and $4K^2P/85$, respectively. Table~\ref{tab:values} reports the normalized robustness score
$\bar\lambda:=\Lambda(\enc;f)/(K^2P)$, whose optimum is
$\bar\lambda^\star:=\Lambda^\star(P;f)/(K^2P)$, and the corresponding
normalized noise-amplification constant
$K\sqrt P\,L_\enc=\bar\lambda^{-1/2}$.

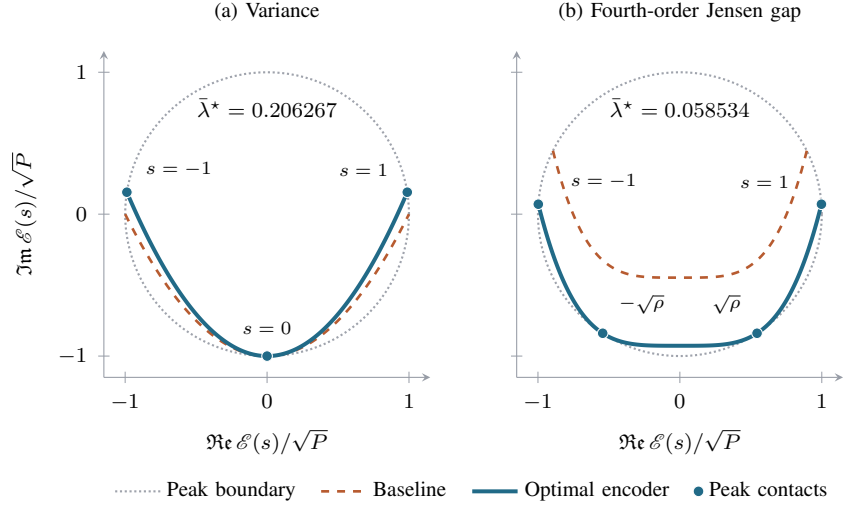
\begin{figure*}[t]
\centering\paperfigure{%
\begin{tikzpicture}
\pgfplotstableread[row sep=\\]{
p M rho rootrho a b c score \\
2 2 0 0 0.98796140786225539 1.1547005383792517 -1 0.2062673845056687 \\
4 4 0.29720785528286564 0.54516773132941909 0.99753915620560607 0.99741350060215672 -0.92730193353503576 0.058533506997747176 \\
2 1 0 0 0.94280904158206336 1.3333333333333333 -1 0.59259259259259256 \\
}\analyticcoefficients
\def\selectanalytic#1{%
  \pgfplotstablegetelem{#1}{a}\of\analyticcoefficients\edef\curveA{\pgfplotsretval}%
  \pgfplotstablegetelem{#1}{b}\of\analyticcoefficients\edef\curveB{\pgfplotsretval}%
  \pgfplotstablegetelem{#1}{c}\of\analyticcoefficients\edef\curveC{\pgfplotsretval}%
  \pgfplotstablegetelem{#1}{rho}\of\analyticcoefficients\edef\curveRho{\pgfplotsretval}%
  \pgfplotstablegetelem{#1}{rootrho}\of\analyticcoefficients\edef\curveRoot{\pgfplotsretval}%
}
\begin{groupplot}[
  group style={group size=2 by 1,horizontal sep=0.45in},
  width=1.70in,height=1.70in,scale only axis,axis equal image,
  xmin=-1.15,xmax=1.15,ymin=-1.15,ymax=1.15,
  xtick={-1,0,1},ytick={-1,0,1},
  axis lines=left,axis line style={figGray!70},tick style={figGray!70},
  /tikz/font=\fontsize{8}{9.5}\selectfont,
  tick label style={font=\fontsize{8}{9.5}\selectfont},
  title style={font=\fontsize{8}{10}\selectfont,yshift=1pt},
  xlabel={$\Re\enc(s)/\sqrt P$},
  xlabel style={yshift=1pt},
  clip=false,
  legend style={draw=none,font=\fontsize{8}{9.5}\selectfont,
    /tikz/every even column/.append style={column sep=7pt}},
  legend columns=4,
]
\nextgroupplot[legend to name=analyticCurveLegend,title={(a) Variance},ylabel={$\Im\enc(s)/\sqrt P$}]
\selectanalytic{0}
\addplot[figGray!65,densely dotted,line width=0.8pt,domain=0:360,samples=121]
  ({cos(x)},{sin(x)});
\addlegendentry{Peak boundary}
\addplot[figOrange,dashed,line width=1.0pt,domain=-1:1,samples=101]
  ({x},{x^2-1});
\addlegendentry{Baseline}
\addplot[figBlue,line width=1.5pt,domain=-1:1,samples=141]
  ({\curveA*x},{\curveB*x^2+\curveC});
\addlegendentry{Optimal encoder}
\addplot[only marks,mark=*,mark size=2.0pt,figBlue,
  mark options={fill=figBlue,draw=white,line width=0.4pt}]
  coordinates {(-\curveA,\curveB-1) (0,-1) (\curveA,\curveB-1)};
\addlegendentry{Peak contacts}
\node[font=\fontsize{8}{9.5}\selectfont] at (axis cs:0,0.75)
  {$\bar\lambda^\star=0.206267$};
\node[anchor=south west,xshift=3pt,yshift=2pt,font=\fontsize{7}{8}\selectfont]
  at (axis cs:-\curveA,\curveB-1) {$s=-1$};
\node[anchor=south east,xshift=-3pt,yshift=2pt,font=\fontsize{7}{8}\selectfont]
  at (axis cs:\curveA,\curveB-1) {$s=1$};
\node[anchor=south,yshift=4pt,font=\fontsize{7}{8}\selectfont]
  at (axis cs:0,-1) {$s=0$};

\nextgroupplot[title={(b) Fourth-order Jensen gap},yticklabels={}]
\selectanalytic{1}
\addplot[figGray!65,densely dotted,line width=0.8pt,domain=0:360,samples=121]
  ({cos(x)},{sin(x)});
\addplot[figOrange,dashed,line width=1.0pt,domain=-1:1,samples=141]
  ({2*x/sqrt(5)},{(2*x^4-1)/sqrt(5)});
\addplot[figBlue,line width=1.5pt,domain=-1:1,samples=181]
  ({\curveA*x},{\curveB*x^4+\curveC});
\addplot[only marks,mark=*,mark size=2.0pt,figBlue,
  mark options={fill=figBlue,draw=white,line width=0.4pt},
  samples at={-1,-\curveRoot,\curveRoot,1}]
  ({\curveA*x},{\curveB*x^4+\curveC});
\node[font=\fontsize{8}{9.5}\selectfont] at (axis cs:0,0.75)
  {$\bar\lambda^\star=0.058534$};
\node[anchor=south west,xshift=8pt,yshift=2pt,font=\fontsize{7}{8}\selectfont]
  at (axis cs:-\curveA,\curveB+\curveC) {$s=-1$};
\node[anchor=south east,xshift=-8pt,yshift=2pt,font=\fontsize{7}{8}\selectfont]
  at (axis cs:\curveA,\curveB+\curveC) {$s=1$};
\node[anchor=south west,xshift=2pt,yshift=3pt,font=\fontsize{7}{8}\selectfont]
  at (axis cs:{-\curveA*sqrt(\curveRho)},{\curveB*\curveRho^2+\curveC})
  {$-\sqrt\rho$};
\node[anchor=south east,xshift=-2pt,yshift=3pt,font=\fontsize{7}{8}\selectfont]
  at (axis cs:{\curveA*sqrt(\curveRho)},{\curveB*\curveRho^2+\curveC})
  {$\sqrt\rho$};

\end{groupplot}
\node[anchor=north,yshift=-0.43in] at ($(group c1r1.south)!0.5!(group c2r1.south)$)
  {\pgfplotslegendfromname{analyticCurveLegend}};
\end{tikzpicture}
% END figures/tikz/optimal_curves
}
\caption{{\color{black} Optimal and baseline encoders for variance and the fourth-order Jensen gap. Markers show the inputs at which the optimal encoders reach peak power. All curves satisfy \eqref{eq:peak-model}; scores are normalized by $K^2P$.}}
\label{fig:curves}\end{figure*}

\begin{table}[t]\caption{Normalized robustness scores.}\label{tab:values}
\centering\small\begin{tabular}{@{}lrrr@{}}\toprule
Desired function & Baseline $\bar\lambda$ & Optimal $\bar\lambda^\star$ & Optimal $(\bar\lambda^\star)^{-1/2}$\\\midrule
$m_2-m_1^2$ & $0.200000$ & $0.206267$ & $2.201835$\\
$m_2-m_1^4$ & $0.047059$ & $0.058534$ & $4.133308$\\
\bottomrule
\end{tabular}\end{table}

For the polynomial encoder
$\enc(s)=\sqrt P[as+i(bs^p+c_0)]$, with $p\in\{2,4\}$,
$a,b>0$, and $c_0\in\R$, projection onto
$\mathcal C=\mathbb I\times\R$ gives
\begin{align}
\hat m_1=\operatorname{clip}_{[-1,1]}
\left(\frac{\Re(y)}{Ka\sqrt P}\right),~~
\hat m_2=\frac{\Im(y)/(K\sqrt P)-c_0}{b},\quad 
\dec_{\mathrm{strip}}(y)=\hat m_2-\psi(\hat m_1). \label{eq:polynomial-decoder}
\end{align}
Here $\operatorname{clip}_{[a,b]}(x):=\min\{b,\max\{a,x\}\}$.
For variance and the fourth-order Jensen gap, we set $\psi(x)=x^p$
and clip the output to $[0,1]$:
$\hat f=\operatorname{clip}_{[0,1]}
(\hat m_2-\hat m_1^p)$.
Figure~\ref{fig:decoder-geometry} illustrates the projection for
variance at $K=4$. Appendix~\ref{app:strip-decoder-proof} derives
the rule and verifies its uniform error bound.

\subsection{Certified Exponential Design}

For $K\ge4$, we next take $\phi(s)=\psi(s)=e^s$, giving the desired function
$$
 f(\boldsymbol s)
 =\frac1K\sum\nolimits_{k=1}^K e^{s_k}
  -e^{K^{-1}\sum_{k=1}^K s_k}.
$$
The exponential functions lack the symmetry used to derive the preceding polynomial designs, so we compute the encoder coefficients numerically. A finite design grid alone cannot guarantee peak-power feasibility between grid points or certify how close the resulting robustness score is to optimal. We therefore use the procedure in Appendix~\ref{sec:certified} to verify the peak-power constraint for every $s\in\mathbb I$ and obtain rigorous lower and upper bounds on $\Lambda^\star(P;f)$.

We solve the normalized problem with $P=1$ on a 401-point
design grid, then round and repair the candidate matrices
to obtain rational primal and dual certificates. We verify
the primal matrix and robustness constraints, together with
the dual conditions \eqref{eq:dualwitness}, using rational
interval arithmetic.\footnote{For rational $s$ with $|s|\le1$, we
enclose $e^s$ using
$|e^s-\sum_{j=0}^{40}s^j/j!|\le3/41!$.
We use certified lower bounds on leading principal minors
for the positive-definiteness checks.} To enforce the peak-power constraint between grid points,
we evaluate the primal power function $q_H$ on a uniform
1001-point verification grid with spacing $h=1/500$.
We obtain $b_h\le1.000000093894$ and
$|q_H''(s)|\le117$ on $\mathbb I$.
The interpolation bound in
Appendix~\ref{app:certificate-witness-bounds} therefore gives
$$
\max_{s\in\mathbb I}q_H(s)
\le b_h+\frac{117h^2}{8}
\le1.000058593894.
$$
We scale the primal matrix and objective by
$\rho=1/1.000058593894$, obtaining the achievable lower bound
$\lambda_{\mathrm{lo}}=\rho\lambda$.
The verified dual certificate gives
$\lambda_{\mathrm{up}}=\sum_i\alpha_i$.
After restoring the power scaling, these certificates yield
\begin{equation}
0.419645089928
\le \frac{\Lambda^\star(P;f)}{K^2P}
\le 0.419670224438.
\label{eq:expcertificate}
\end{equation}

We reconstruct the encoder from the scaled primal Gram
matrix and multiply its coefficients by $\sqrt P$.
Its coordinates are approximately
\begin{align*}
\frac{\Re\enc(s)}{\sqrt P}
&\approx1.8321854+2.1939971s-1.1944545e^s,\\
\frac{\Im\enc(s)}{\sqrt P}
&\approx-0.7900289+0.5211418e^s.
\end{align*}
The certified bounds apply to the encoder reconstructed from the verified rational matrix; the displayed coefficients are rounded approximations. Figure~\ref{fig:certified} illustrates the resulting encoder trajectory alongside the decay of the certified relative optimality gap as the verification grid is refined.

\begin{figure*}[t]
 \centering\paperfigure{%
% BEGIN figures/tikz/certified_synthesis
% Native exponential curve and five exact-certificate summary values.
\begin{tikzpicture}
\pgfplotstableread[row sep=\\]{
c1 c2 l11 l12 l22 \\
1.8321854032246792 -0.79002894450391214 2.1939971000776324 -1.194454497024295 0.52114176864689676 \\
}\expcoefficients
\pgfplotstablegetelem{0}{c1}\of\expcoefficients\edef\expCOne{\pgfplotsretval}
\pgfplotstablegetelem{0}{c2}\of\expcoefficients\edef\expCTwo{\pgfplotsretval}
\pgfplotstablegetelem{0}{l11}\of\expcoefficients\edef\expLOne{\pgfplotsretval}
\pgfplotstablegetelem{0}{l12}\of\expcoefficients\edef\expLTwo{\pgfplotsretval}
\pgfplotstablegetelem{0}{l22}\of\expcoefficients\edef\expLThree{\pgfplotsretval}
\begin{groupplot}[
  group style={group size=2 by 1,horizontal sep=0.57in},
  width=2.65in,height=1.70in,scale only axis,
  axis lines=left,axis line style={figGray!70},tick style={figGray!70},
  /tikz/font=\fontsize{8}{9.5}\selectfont,
  tick label style={font=\fontsize{8}{9.5}\selectfont},
  title style={font=\fontsize{8}{10}\selectfont,yshift=2pt},
  legend style={draw=none,font=\fontsize{8}{9.5}\selectfont,
    cells={anchor=west}},
]
\nextgroupplot[
  title={(a) Certified exponential encoder},axis equal image,
  xmin=-1.15,xmax=1.15,ymin=-1.15,ymax=1.15,
  xtick={-1,0,1},ytick={-1,0,1},
  xlabel={$\Re(\enc(s))/\sqrt P$},
  ylabel={$\Im(\enc(s))/\sqrt P$},
  legend to name=certifiedCurveLegend,legend columns=3,
  legend style={/tikz/every even column/.append style={column sep=6pt}},
  clip=false,
]
\addplot[figGray!65,densely dotted,line width=0.8pt,domain=0:360,samples=141]
  ({cos(x)},{sin(x)});
\addlegendentry{Peak boundary}
\addplot[figBlue,line width=1.5pt,domain=-1:1,samples=181]
  ({\expCOne+\expLOne*x+\expLTwo*exp(x)},{\expCTwo+\expLThree*exp(x)});
\addlegendentry{Feasible encoder}
\addplot[figBlue,only marks,mark=o,mark size=2.2pt,
  mark options={fill=white,line width=0.9pt},samples at={-1,0.795,1}]
  ({\expCOne+\expLOne*x+\expLTwo*exp(x)},{\expCTwo+\expLThree*exp(x)});
\addlegendentry{Candidate contacts}
\node[font=\fontsize{7}{8}\selectfont,anchor=south east,xshift=24pt,yshift=6pt]
  at (axis cs:-0.8012269497,-0.5983116019) {$s=-1$};
\node[font=\fontsize{7}{8}\selectfont,anchor=north east,xshift=-4pt,yshift=-1pt]
  at (axis cs:0.9313640908,0.3640087531) {$s=0.795$};
\node[font=\fontsize{7}{8}\selectfont,anchor=south east,xshift=-4pt,yshift=-2pt]
  at (axis cs:0.7793185491,0.6265812553) {$s=1$};

\nextgroupplot[
  title={(b) Verification-grid refinement},
  xmode=log,ymode=log,xmin=43,xmax=1300,ymin=0.004,ymax=4,
  xtick={51,101,201,401,1001},xticklabels={51,101,201,401,1001},
  ytick={0.01,0.1,1},yticklabels={$0.01$,$0.1$,$1$},
  xlabel={Number of verification grid points},ylabel={Certified relative gap (\%)},
  grid=major,major grid style={figGray!18,line width=0.3pt},
  minor tick num=0,
]
\addplot[figBlue,line width=1.3pt,mark=*,mark size=2.1pt,
  mark options={fill=figBlue,draw=white,line width=0.4pt}]
  table[x=n,y=gap_pct,row sep=\\]{
n gap_pct lower upper \\
51 2.2866320556722495 0.41007391055747533 0.41967022443757662 \\
101 0.58173624879899311 0.41722885061660714 0.41967022443757662 \\
201 0.14617566484116248 0.41905676869686459 0.41967022443757662 \\
401 0.036688542782721593 0.41951625354773747 0.41967022443757662 \\
1001 0.005989109550760389 0.41964508992808314 0.41967022443757662 \\
};
\node[anchor=north east,font=\fontsize{8}{9.5}\selectfont,
  fill=white,inner sep=2pt,text=figBlue]
  at (rel axis cs:0.97,0.95) {$100(U-L)/U$};
\node[anchor=south east,font=\fontsize{8}{9.5}\selectfont,
  fill=white,inner sep=2pt,text=figBlue]
  at (axis cs:1015,0.0103) {$0.005989\%$};
\end{groupplot}
\node[anchor=north,yshift=-0.44in] at ($(group c1r1.south west)!.5!(group c2r1.south east)$)
  {\pgfplotslegendfromname{certifiedCurveLegend}};
\end{tikzpicture}
% END figures/tikz/certified_synthesis
}
 \caption{{\color{black} Certified exponential encoder. Left: the reconstructed encoder, unit peak boundary, and candidate contact inputs. Right: the certified relative gap under grid refinement, with the candidate and dual certificate fixed. The repaired encoder need not touch the peak boundary.}}
 \label{fig:certified}
\end{figure*}
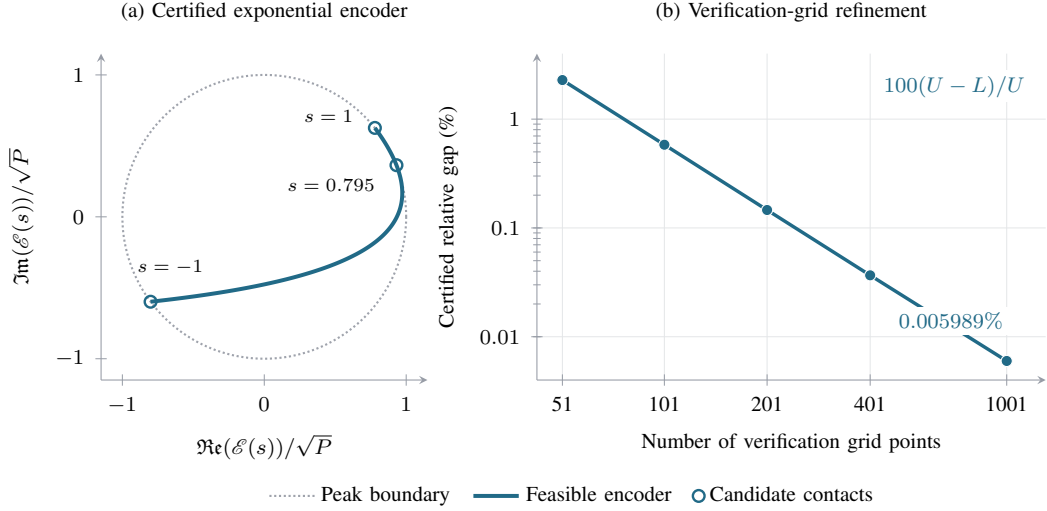

\subsection{MSE under a Robustness Requirement}\label{sec:pareto-evidence}

For variance $V$, we optimize $b$ within the following encoder family, shared by all transmitters and satisfying the peak-power constraint for every $1\le b<2$:
$$
 \enc_b(s)=\sqrt P\left[\sqrt{b(2-b)}\,s+i(bs^2-1)\right],
 \qquad 1\le b<2,
$$
using the polynomial decoder $\dec_{\mathrm{strip}}$ from \eqref{eq:polynomial-decoder} with outputs clipped to $[0,1]$. Under this parameterization, the Gram matrix is $G=P\operatorname{diag}(b(2-b),b^2)$. The normalized robustness score follows from
Lemma~\ref{prop:score}, while the high-SNR MSE coefficient
follows from \eqref{eq:small-noise-mse} with $p=2$ and
$a^2=b(2-b)$:
$$
\bar\lambda(b):=\frac{{\Lambda(\enc_b;V)}}{K^2P}
 =\left[\frac{4}{b(2-b)}+\frac1{b^2}\right]^{-1},\quad
 \mathcal J_\pi(G)=\frac1{2K^2P}
 \left[\frac{4\mathbb E_\pi[m_1^2]}{b(2-b)}+\frac1{b^2}\right].
$$
We select $b$ to minimize $\mathcal J_\pi(G)$ subject to
$\bar\lambda(b)\ge\eta\bar\lambda(2/\sqrt3)$, where
$\eta\in[0,1]$ is the required fraction of the maximum robustness score. The encoder with $b=2/\sqrt3$ maximizes robustness. For $\eta\in(0,1]$, retaining a fraction $\eta$ of the maximum robustness score limits the uniform noise-amplification constant $L_\enc$ to at most $\eta^{-1/2}$ times its minimum value.

For the numerical evaluations, we set $P=1$ and $K\in\{4,16\}$.
We draw independent inputs either from the uniform distribution on $[-1,1]$ or
from the endpoint-Beta distribution, defined by
$(s+1)/2\sim\operatorname{Beta}(0.3,0.3)$. These input distributions induce the moment distribution $\pi$, with $\mathbb E_\pi[m_1^2]=1/(3K)$ and $5/(8K)$, respectively. Figure~\ref{fig:pareto} compares the relative reductions in the high-SNR MSE coefficient $\mathcal J_\pi(G)$ for different
encoder coefficients and illustrates the peak constraint~\eqref{eq:peak-model} through an enclosing ellipse~\cite{boyd2004}. For a shared encoder with invertible
$\mathsf U$, let $\boldsymbol a=-\mathsf U^{-1}\boldsymbol c$.
Its peak constraint holds precisely when
$$
\bigl((\phi_1(s),\phi_2(s))^\T-\boldsymbol a\bigr)^\T
G\bigl((\phi_1(s),\phi_2(s))^\T-\boldsymbol a\bigr)
\le P,\qquad s\in\mathbb I.
$$
This ellipse is centered at $\boldsymbol a$, and the quadratic
form equals $|\enc(s)|^2$. Conversely, any such enclosing ellipse
with $G\succ0$ gives a peak-feasible encoder by choosing
$\mathsf U^\T\mathsf U=G$ and
$\boldsymbol c=-\mathsf U\boldsymbol a$. Thus, relaxing the robustness requirement lets us tune $b$ to the input distribution and reduce the high-SNR average-MSE coefficient while retaining at least an $\eta$ fraction of the optimal robustness score.

\input{figures/robustness_mse.tex}

For uniform inputs, minimizing $\mathcal J_\pi(G)$ with $\eta=0$
retains $83.91\%$ and $60.91\%$ of the maximum robustness score
at $K=4$ and $K=16$, respectively. Requiring $90\%$ retention
selects $a=0.9038538$, $b=1.4278414$, and $c_0=-1$ in both cases.
Relative to the robustness-optimal encoder, this reduces the
high-SNR MSE coefficient by $17.68\%$ and $29.07\%$, respectively,
while increasing the optimal uniform noise-amplification constant $L_\enc$
by $5.41\%$. For endpoint-Beta inputs with $K=4$, the choice with $\eta=0$ retains $91.67\%$ of the maximum score, so the $90\%$
requirement is inactive.

We also evaluate eight retention levels at
$P/\sigma_z^2=20$ dB using $200{,}000$ independent input and noise
draws shared across encoders. Table~\ref{tab:pareto} compares the
predicted and measured MSE reductions for $90\%$ retention.
We compute nominal pointwise $95\%$ confidence intervals by the
paired delta method, accounting for the shared draws. 

%By retaining an $\eta$ fraction of optimal robustness, with
%$\eta>0$, we limit uniform noise amplification to at most
%$\eta^{-1/2}$ times its minimum value. 

\begin{table}[t]
\caption{MSE reduction (\%) relative to the robustness-optimal variance encoder, with at least $90\%$ of the maximum robustness score retained. U: uniform; B: endpoint-Beta.}
\label{tab:pareto}
\centering\small
\begin{tabular}{@{}ccrr@{}}\toprule
Inputs & $K$ & High-SNR limit & $20$ dB, paired $95\%$ interval\\\midrule
U & 4 & 17.68 & $17.99\ [17.84,18.14]$\\
U & 16 & 29.07 & $29.10\ [29.02,29.19]$\\
B & 4 & 9.79 & $10.85\ [10.70,11.00]$\\
B & 16 & 25.09 & $25.12\ [25.00,25.24]$\\\bottomrule
\end{tabular}
\end{table}

\subsection{Finite-SNR Performance}\label{sec:finite-noise}

We compare the average MSE $R_\pi(\enc,\dec;\sigma_z^2)$ across SNR values for the robustness-optimal encoders, the baselines in Figure~\ref{fig:curves}, and prior-aware encoders. We design the prior-aware encoders using the input distribution to minimize the high-SNR average-MSE coefficient $\mathcal J_\pi(G)$.

For $p\in\{2,4\}$, write $f_p:=f_{\chi_p,\chi_p}=m_2-m_1^p$, where $m_2=K^{-1}\sum_k s_k^p$ and $f_2=V$. We use the degree-$p$ polynomial encoder family introduced in Section~\ref{sec:quadratic-quartic} and decode using \eqref{eq:polynomial-decoder}, followed by output clipping to $[0,1]$. For independent inputs with a common continuous density on $(-1,1)$, we obtain
\begin{equation}
 \lim_{\sigma_z^2\downarrow0}
 \frac{R_\pi(\enc,\dec;\sigma_z^2)}{\sigma_z^2}
 =\frac{1}{2K^2P}\left(\frac{1}{b^2}
 +\frac{p^2\mathbb E_\pi[m_1^{2p-2}]}{a^2}\right).
 \label{eq:small-noise-mse}
\end{equation}
We obtain the limit in \eqref{eq:small-noise-mse} from the local expansion of the polynomial decoder and its uniform error bound.\footnote{Under the stated
input assumptions, $|m_1|<1$ and $0<f_p<1$ almost surely, so
the strip projection and output clipping are locally inactive.
The uniform error bound in
Appendix~\ref{app:strip-decoder-proof} permits dominated
convergence, following Step~1 of
Appendix~\ref{genstat:information}. This argument does not
require a density for the induced moment distribution $\pi$.}
The two terms correspond to errors in estimating $m_2$ and $m_1$, respectively. The input distribution determines their relative weights. 

For the orthogonal polynomial encoders, whose coefficient
vectors satisfy $\boldsymbol u^\T\boldsymbol v=0$, set
$A=Pa^2$, $B=Pb^2$, and
$M_{\rm eff}=p\sqrt{\mathbb E_\pi[m_1^{2p-2}]}>0$.
Equation~\eqref{eq:small-noise-mse} then gives
$2K^2\mathcal J_\pi(G)=\frac{M_{\rm eff}^2}{A}+\frac1B.$
This objective has the same form as the robustness
optimization in Proposition~\ref{thm:parabola}, with
$M$ replaced by $M_{\rm eff}$ and unchanged peak-power
constraints. We therefore obtain the prior-aware
coefficients from the same quadratic and quartic formulas.
We compute these moments analytically and use the decoding rule~\eqref{eq:polynomial-decoder} with output clipping for all designs. These encoder coefficients minimize the high-SNR coefficient within the family, but need not minimize finite-SNR MSE.

We retain $P=1$, $K\in\{4,16\}$, and the uniform and endpoint-Beta input distributions. For each input distribution and each value of $K$, we generate $200{,}000$ independent input vectors and standardized Gaussian noise samples. We share these draws across encoders and desired functions, scaling the noise to test $P/\sigma_z^2$ from $-10$ to $30$ dB in $5$-dB increments. Paired standard errors account for the shared draws; the reported confidence intervals are pointwise.

\input{figures/finite_noise.tex}

\begin{table}[t]
\caption{MSE at $P/\sigma_z^2=20$ dB, in units of $10^{-4}$. U: uniform; B: endpoint-Beta. Robust: robustness-optimal encoder; prior-aware: encoder selected to minimize the high-SNR MSE coefficient within the orthogonal family.}
\label{tab:finite-noise}
\centering\small
\setlength{\tabcolsep}{5pt}
\begin{tabular}{@{}ccrrrrrr@{}}\toprule
 & & \multicolumn{3}{c}{$p=2$} & \multicolumn{3}{c}{$p=4$}\\
\cmidrule(lr){3-5}\cmidrule(l){6-8}
Inputs & $K$ & Baseline & Robust & Prior-aware & Baseline & Robust & Prior-aware\\\midrule
U & 4  & 4.055 & 3.318 & 2.684 & 4.137 & 3.341 & 1.652\\
B & 4  & 4.787 & 4.075 & 3.634 & 5.627 & 4.548 & 3.338\\
U & 16 & 0.2115 & 0.1630 & 0.09922 & 0.2450 & 0.1970 & 0.05521\\
B & 16 & 0.2252 & 0.1772 & 0.1232 & 0.2471 & 0.1987 & 0.06590\\\bottomrule
\end{tabular}
\end{table}

Figure~\ref{fig:finite-noise} and Table~\ref{tab:finite-noise}
show MSE reductions of $6.12\%$--$22.92\%$ for the robustness-optimal
encoders relative to the baselines across the $72$ tested
combinations of desired function, input distribution, $K$, and SNR. Prior-aware tuning gives a further $5.69\%$--$72.01\%$
reduction, but can substantially reduce robustness. For the
fourth-order Jensen gap with uniform inputs and $K=16$, the
normalized robustness score $\Lambda(\enc;f_4)/(K^2P)$ decreases from $0.05853$ to $0.007273$.

At low SNR, a decoder based only on the input distribution can
outperform the implemented decoder. For variance with uniform
inputs and $K=4$, the constant prior-mean estimator has exact MSE $11/480\approx0.02292$, independently of SNR. At $-10$ dB, the robustness-optimal encoder with the clipped projection decoder yields an empirical MSE of $0.1032$. Thus, the gains over the
baselines do not establish finite-SNR Bayesian optimality.

Clipping also affects convergence to the high-SNR limit. At
$30$ dB, the ratios $R_\pi/\sigma_z^2$ for endpoint-Beta inputs
with $K=4$ remain $3.49\%$--$7.19\%$ below their limiting values.
These differences exceed the Monte Carlo uncertainty, showing
that the asymptotic coefficient alone does not capture the
observed performance at this SNR. Nevertheless, Figure~\ref{fig:finite-noise} demonstrates that prior-aware encoder tuning based on the high-SNR coefficient continues to yield substantial MSE reductions across moderate and high SNRs, even where the asymptotic ratio has not yet converged.

%Moreover,  the experiments show that high-SNR encoder tuning can improve MSE even when the asymptotic approximation is inaccurate. 

\section{Conclusions}\label{sec:boundaries}

We have developed a geometric and optimization framework for over-the-air computation of functions of two additive statistics using continuous encoders $\enc_k$ and decoders $\dec$ under pointwise peak-power constraints. Our characterization establishes when exact noiseless computation forces every continuous encoder collection to preserve both statistics $m_1,m_2$, reducing infinite-dimensional encoder design to optimization over a common finite-dimensional representation. Within this family, we formulated robustness maximization as an exact convex semi-infinite program—eliminating rank constraints without a relaxation gap—and constructed a Lipschitz-extension decoder that achieves the minimum uniform noise-amplification constant whenever $\Lambda(\enc;f)>0$. For continuously differentiable postprocessing functions $\psi_1,\psi_2$ and any fixed encoder collection with an invertible coefficient matrix $\mathsf U$, this same decoder attains the minimum leading average Gaussian MSE coefficient for every regular moment distribution. For worst-case Gaussian MSE, we further proved that a single encoder--decoder pair independent of the noise power achieves the optimal asymptotic minimax coefficient over $\mathcal B_P$ under the conditions of Theorem~\ref{thm:identifiable-minimax}.

Our closed-form polynomial designs and numerically certified exponential encoder illustrate both analytical and computational synthesis within this framework. Moreover, our numerical evaluations demonstrate how prior-aware tuning can substantially reduce average MSE $R_\pi$ at the expense of uniform robustness $\Lambda(\enc;f)$, highlighting that the optimal encoder geometry depends fundamentally on the chosen error criterion. 

Several natural directions remain for future investigation. On the structural side, it is of interest to determine whether the affine encoder characterization remains necessary for $K=3$, to generalize the framework to functions of more than two aggregate statistics across multiple channel uses or subcarriers, and to relax continuity to piecewise-continuous encoders that bridge analog mappings with quantized constellations. On the estimation side, our low-SNR comparisons motivate jointly optimizing encoders and prior-aware decoders for finite-SNR average MSE subject to an explicit uniform robustness guarantee. Finally, connecting the deterministic robustness score to finite-SNR information-theoretic distortion bounds could further clarify how worst-case noise resilience constrains average performance near fundamental limits.

\section{Acknowledgment}

The author also acknowledges using ChatGPT to assist with manuscript writing, code development, and proof assistance. ChatGPT also assisted with the Lean formalization and verifying  the proofs of all theorems. 

\appendices

\section{Encoding Proofs}

\subsection{Proof of Theorem~\ref{thm:rigidity}}
\label{app:rigidity-proof}

We prove necessity in three steps. We first use two choices
of a fourth-node input to establish local recovery of both
partial moment sums. We then show that the aggregate depends
only on these sums and establish encoder differentiability.
Finally, we derive common coefficients across all nodes.
Sufficiency follows by inverting the resulting affine map.

Assume that both $\psi_1$ and $\psi_2$ are nonconstant and at least one is
nonaffine. By the assumption on affine subintervals, neither $\psi_j'$ vanishes identically on an open interval.
Write $[a_{j,k},b_{j,k}]=\phi_{j,k}(\mathbb I)$; independence of
$1,\phi_{1,k},\phi_{2,k}$ gives $a_{j,k}<b_{j,k}$.
Choose any three nodes and a fourth, relabeled $1,2,3,4$.
Fix the remaining inputs and put
$d_j=\sum_{k=5}^K\phi_{j,k}(s_k^*)$, with an empty sum equal to zero.
For the partial sums $p_j=\sum_{k=1}^3\phi_{j,k}(s_k)$ and
input $t$ at node $4$, the desired function value is
 \begin{align*}
 C_t(\boldsymbol p)=\sum\nolimits_{j=1}^2
 \psi_j\!\left(\frac{p_j+d_j+\phi_{j,4}(t)}K\right),
 \qquad t\in\mathbb I,~ 
 \boldsymbol p\in
 \prod_{j=1}^2\left(\sum_{k=1}^3a_{j,k},\sum_{k=1}^3b_{j,k}\right).
 \end{align*}
The rectangle may include unattainable partial moment vectors $\boldsymbol p$. Nevertheless, every argument of $\psi_j$ lies strictly inside $\mathbb J_j$,
so $C_t$ is $C^1$ throughout the rectangle.

\emph{Step 1: Two desired function values locally determine both partial sums.}
Let $\mathcal O$ be the set of points where
$\det[\nabla C_t\ \nabla C_u]\ne0$ for some $t,u\in\mathbb I$.
Continuity makes $\mathcal O$ open. We show that it is dense by
contradiction. Otherwise, an open rectangle contains no point of
$\mathcal O$, so every gradient determinant vanishes there.
After translating and scaling the coordinates, this gives
%------------
\begin{align*}
     A(\xi+x_t)B(\eta+y_u)=A(\xi+x_u)B(\eta+y_t)
 \quad(\xi\in U,\ \eta\in V,\ t,u\in\mathbb I),
\end{align*}
%------------
where $A=\psi_1'$, $B=\psi_2'$, $x_t=\phi_{1,4}(t)/K$,
$y_t=\phi_{2,4}(t)/K$, and $U,V$ are nonempty open intervals.

Fix $u$. Choose $\xi_0\in U$, $\eta_0\in V$ with
$A(\xi_0+x_u)B(\eta_0+y_u)\ne0$; neither derivative vanishes
on an open interval. Evaluating the determinant identity at these
points defines the common ratio
$$
r(t):=\frac{A(\xi_0+x_t)}{A(\xi_0+x_u)}
=\frac{B(\eta_0+y_t)}{B(\eta_0+y_u)}.
$$
Setting only $\eta=\eta_0$ or $\xi=\xi_0$ then gives
$$
 A(\xi+x_t)=r(t)A(\xi+x_u),\qquad
 B(\eta+y_t)=r(t)B(\eta+y_u).
$$
If $r(t)=0$, the first identity makes $A$ vanish on an open
interval. Thus $r(t)\ne0$; continuity and $r(u)=1$ imply $r(t)>0$.
Likewise, $A(\xi+x_u)=0$ would force
$A(\xi+x_t)=0$ for every $t$, so $A$ would vanish on the
nondegenerate interval $\xi+\phi_{1,4}(\mathbb I)/K$.
Hence $A$ is nonzero at every argument above, as is $B$.

Put $X=\phi_{1,4}(\mathbb I)/K$. Since $r(t)$ depends only on $x_t$
and $A\ne0$, taking logarithms gives
$\log|A(\xi+x)|=h(\xi)+k(x)$ for $(\xi,x)\in U\times X$,
with $h,k$ continuous.

For $x\in X^\circ$ and sufficiently small $\delta$, the arguments
$(\xi+\delta)+x$ and $\xi+(x+\delta)$ coincide. Therefore
$$
h(\xi+\delta)-h(\xi)=k(x+\delta)-k(x)=:D(\delta).
$$
The left side is independent of $x$, and the right side is independent
of $\xi$. These common increments are continuous and satisfy
$D(\delta+\varepsilon)=D(\delta)+D(\varepsilon)$ for sufficiently
small increments. Rational subdivision and continuity give
$D(\delta)=\lambda_1\delta$. Subdividing longer intervals then gives
$h(\xi)=\lambda_1\xi+h_0$ and $k(x)=\lambda_1x+k_0$;
continuity covers the endpoints of $X$. Hence $|A|$ is exponential
on $U+X$. Applying the same argument to $B$ yields
$$
r(t)=e^{\lambda_1(x_t-x_u)}=e^{\lambda_2(y_t-y_u)},
\qquad
\lambda_1\phi_{1,4}(t)-\lambda_2\phi_{2,4}(t)
=\lambda_1\phi_{1,4}(u)-\lambda_2\phi_{2,4}(u).
$$
Independence of $1,\phi_{1,4},\phi_{2,4}$ forces
$\lambda_1=\lambda_2=0$. Thus $A=\psi_1'$ and $B=\psi_2'$
are constant on open intervals, making both $\psi_1$ and $\psi_2$
affine there. This contradicts the assumption that at least one
is nonaffine on every open subinterval.

Thus $\mathcal O$ is dense. At each of its points, the inverse-function
theorem gives a neighborhood where
$\boldsymbol p\mapsto(C_t(\boldsymbol p),C_u(\boldsymbol p))$ is one-to-one.
The choices $t,u$ compare alternative input vectors in the proof;
they do not require additional channel uses.

\emph{Step 2: The aggregate depends only on the partial sums, and the encoders are differentiable.}

We represent each complex encoder value $\mathscr E_k(s)$ by its
in-phase and quadrature components in $\R^2$. For each node $k$, the family $\{ (\phi_{1,k}'(s),\phi_{2,k}'(s))^\T:
s\in\mathbb I^\circ \}$ spans $\R^2$.
Otherwise, some $(\alpha,\beta)\ne(0,0)$ would satisfy
$\alpha\phi_{1,k}'+\beta\phi_{2,k}'=0$, giving
$\alpha\phi_{1,k}+\beta\phi_{2,k}=\text{constant}$ by integration.
This contradicts independence of $1,\phi_{1,k},\phi_{2,k}$.
We can therefore choose $a,b\in\mathbb I^\circ$ at nodes $2,3$ with
$$
 \det\begin{pmatrix}\phi_{1,2}'(a)&\phi_{1,3}'(b)\\
                         \phi_{2,2}'(a)&\phi_{2,3}'(b)\end{pmatrix}\ne0.
$$
Fix any $s_1=s_0\in\mathbb I^\circ$.
The map $(s_2,s_3)\mapsto(p_1,p_2)$ has the displayed nonsingular
Jacobian at $(a,b)$, so it has a local $C^1$ inverse.
Its image contains an open neighborhood of partial moment vectors
inside the rectangle of Step~1: the nonzero derivative of each
coordinate excludes its global extrema.

By density of $\mathcal O$, we can adjust $a,b$ so that the partial moment vector lies in $\mathcal O$, while preserving the nonsingular Jacobian. Choose an open neighborhood $B\subset\mathcal O$ where the corresponding $C_t,C_u$ jointly determine $\boldsymbol p$.

Suppose two nearby input triples have partial moment vectors
$\boldsymbol p,\boldsymbol p'\in B$ and equal encoded sums.
Complete both triples with the same remaining inputs, using first
$t$ and then $u$ at node $4$. Identifiability in \eqref{eq:exactness}
and the local one-to-one property from Step~1 give
$$
\sum_{i=1}^3\mathscr E_i(s_i)=\sum_{i=1}^3\mathscr E_i(s_i')
\ \Longrightarrow\ 
\begin{pmatrix}C_t(\boldsymbol p)\\C_u(\boldsymbol p)\end{pmatrix}
=\begin{pmatrix}C_t(\boldsymbol p')\\C_u(\boldsymbol p')\end{pmatrix}
\ \Longrightarrow\ \boldsymbol p=\boldsymbol p'.
$$
Thus, the encoded sum locally determines both $p_1$ and $p_2$, although
$\dec_0$ need only recover the scalar function $f$.

For the reverse dependence, use the local coordinates
$(s_1,s_2,s_3)\mapsto(s_1,p_1,p_2)$.
The same nonsingular Jacobian gives a $C^1$ inverse near
$(s_0,a,b)$. Restrict $(s_1,\boldsymbol p)$ to a product of an
interval and an open subset of $B$. Varying $s_1$ at fixed
$\boldsymbol p$ then adjusts $s_2,s_3$ to preserve both partial sums.

Write the continuous encoded sum in these coordinates as
$\mathscr R(\boldsymbol p,s_1)$. The collision argument defines an auxiliary moment-recovery map $\kappa$ with
$\kappa(\mathscr R(\boldsymbol p,s_1))=\boldsymbol p$; it is distinct from the function decoder $\dec_0$. For each fixed $\tau$, the recovery identity makes
$\mathscr R_\tau:=\mathscr R(\cdot,\tau)$ injective.
By invariance of domain~\cite[Theorem~2B.3]{hatcher},
its image is open. For fixed $\boldsymbol p$ and
$s_1$ sufficiently close to $\tau$, continuity therefore gives
$\mathscr R(\boldsymbol p,s_1)=\mathscr R(\boldsymbol q,\tau)$
for some $\boldsymbol q$ in the chosen moment neighborhood.
Applying $\kappa$ yields $\boldsymbol q=\boldsymbol p$.
Thus, $\mathscr R(\boldsymbol p,\cdot)$ is locally constant
and hence constant on the input interval.
This argument requires no regularity of $\kappa$. 

Thus, a continuous map $\mathscr H$ satisfies,
on a product of sufficiently small input intervals,
\begin{equation}
 \sum\nolimits_{i=1}^3\mathscr E_i(s_i)
 =\mathscr H\!\left(\sum\nolimits_{i=1}^3\phi_{1,i}(s_i),
                         \sum\nolimits_{i=1}^3\phi_{2,i}(s_i)\right).
 \label{eq:localfactor}
\end{equation}
Equation~\eqref{eq:localfactor} is a restricted additive functional equation, related to~\cite[Theorem~1.6]{charalambides2013}. Since our preprocessing curves may contain straight-line portions, we establish the required encoder structure directly.
We establish differentiability of $\mathscr E_k$ before taking any encoder
derivatives. The pair map
$(s_2,s_3)\mapsto\boldsymbol\omega=
(\phi_{1,2}(s_2)+\phi_{1,3}(s_3),
 \phi_{2,2}(s_2)+\phi_{2,3}(s_3))^\T$
has a local $C^1$ inverse. Choose a nonnegative smooth weight $\varrho$ with integral one, supported strictly inside its image. Averaging \eqref{eq:localfactor} over $\boldsymbol\omega$ gives
$$
 \mathscr E_1(s_1)=\int
 \mathscr H\bigl((\phi_{1,1}(s_1),\phi_{2,1}(s_1))^\T+\boldsymbol\omega\bigr)
 \varrho(\boldsymbol\omega)\,d\boldsymbol\omega-\boldsymbol C_0,
$$
where $\boldsymbol C_0$ averages the encoded symbols at nodes $2,3$ and
does not depend on $s_1$. Shrink the input interval and the support of
$\varrho$ so that all arguments of $\mathscr H$ stay in a compact subset
of its domain. Multiply $\mathscr H$ by a smooth cutoff supported inside its domain
and equal to one near this subset. Extend the product by zero to $\R^2$.

Denote this extension by $\widetilde{\mathscr H}$. After changing
variables to $\boldsymbol v=\boldsymbol x+\boldsymbol\omega$, the integral becomes
$$
F(\boldsymbol x)=\int\widetilde{\mathscr H}(\boldsymbol v)
\varrho(\boldsymbol v-\boldsymbol x)\,d\boldsymbol v,
\qquad
\partial_{x_j}F(\boldsymbol x)=-\int
\widetilde{\mathscr H}(\boldsymbol v)
(\partial_j\varrho)(\boldsymbol v-\boldsymbol x)\,d\boldsymbol v.
$$
Derivatives act on the smooth weight, so $F$ is smooth.
Composing $F$ with the $C^1$ preprocessing curve makes
$\mathscr E_1$ continuously differentiable near $s_0$.
Arbitrary $s_0$ and relabeling give $\mathscr E_k\in C^1(\mathbb I^\circ)$
for every node. This averaging does not change the transmission scheme.

\emph{Step 3: All encoders use the same two moment coefficients.} Treat the signal components separately, writing $e_k$ for either real
component of $\mathscr E_k$. Consider three nodes whose partial moment
vector lies in $\mathcal O$, with two independent moment-derivative columns.
By \eqref{eq:localfactor}, their encoded sum stays constant under local
input changes that preserve both partial moment sums.

Let $\mathsf J$ be the $2\times3$ Jacobian of the partial moment
map and let $\boldsymbol q^\T=(e_1'(s),e_2'(a),e_3'(b))$.
Since $\mathsf J$ has rank two, the implicit-function theorem
realizes every $\boldsymbol h$ with $\mathsf J\boldsymbol h=0$
as the tangent to a curve with constant partial moments.
Equation~\eqref{eq:localfactor} makes the encoded sum constant
on this curve, so $\mathsf J\boldsymbol h=0$ implies
$\boldsymbol q^\T\boldsymbol h=0$.
Thus, $\boldsymbol q^\T$ lies in the row span of $\mathsf J$,
and
$$
 \det\begin{pmatrix}
 \phi_{1,1}'(s)&\phi_{1,2}'(a)&\phi_{1,3}'(b)\\
 \phi_{2,1}'(s)&\phi_{2,2}'(a)&\phi_{2,3}'(b)\\
 e_1'(s)&e_2'(a)&e_3'(b)
 \end{pmatrix}=0.
$$
The determinant also vanishes outside $\mathcal O$.
For fixed $s$ and independent reference columns at $(a,b)$,
local invertibility and density give pairs converging to $(a,b)$
whose partial moment vectors lie in $\mathcal O$.
Their determinants are zero; continuity of the derivatives from
Step~2 implies that the limiting determinant is also zero. The argument applies to any three
nodes, using a fourth as in Step~1.

Fix inputs $s_1^*,s_2^*$ with independent moment-derivative
columns. They determine unique $\alpha,\beta$ through
$e_i'(s_i^*)=\alpha\phi_{1,i}'(s_i^*)+\beta\phi_{2,i}'(s_i^*)$,
$i=1,2$. For any other node $k$, subtract $\alpha$ times the first
row and $\beta$ times the second from the last row of the
three-node determinant. Its first two entries vanish and their
complementary $2\times2$ minor is nonzero, giving
$$
 e_k'(s)=\alpha\phi_{1,k}'(s)+\beta\phi_{2,k}'(s)
 \qquad(s\in\mathbb I^\circ).
$$
Choose independent reference columns at nodes $3,4$.
Their encoder derivatives already satisfy this identity, so they
determine the same $\alpha,\beta$. Using them in the determinant
extends the identity to all inputs at nodes $1,2$.
Integration gives
$e_k(s)=\gamma_k+\alpha\phi_{1,k}(s)+\beta\phi_{2,k}(s)$,
with only $\gamma_k$ depending on the node. Applying this result
to both signal components and using endpoint continuity gives
\eqref{eq:affine}.

Finally, $u,v$ must be linearly independent over $\R$. Otherwise,
some $\boldsymbol h\ne0$ satisfies $uh_1+vh_2=0$. Choose
$\boldsymbol p$ in an attainable open neighborhood from Step~2.
For small $\varepsilon\ne0$, both $\boldsymbol p$ and
$\boldsymbol p+\varepsilon\boldsymbol h$ lie there, but
$u(p_1+\varepsilon h_1)+v(p_2+\varepsilon h_2)=up_1+vp_2$.
These equal encoded sums contradict local recovery of $\boldsymbol p$. Conversely, set $\mathsf U=[\boldsymbol u\ \boldsymbol v]$.
Real linear independence makes $\mathsf U$ invertible, so the
receiver recovers both moments from
$$
\boldsymbol m=\frac1K\mathsf U^{-1}
\left(\boldsymbol r_\enc(\boldsymbol s)-\sum_{k=1}^K\boldsymbol c_k\right).
$$
It then evaluates $\psi_1(m_1)+\psi_2(m_2)$, proving sufficiency.

\subsubsection{Excluded Cases} 
The excluded cases need only one real aggregate. If $\psi_1$ is
constant, real encoders proportional to $\phi_{2,k}$ suffice;
the other case is analogous. If $\psi_j(v)=\alpha_jv+\beta_j$
for both $j$, real encoders proportional to
$\alpha_1\phi_{1,k}+\alpha_2\phi_{2,k}$ suffice.
Boundedness of the preprocessing functions permits a common
nonzero scaling satisfying \eqref{eq:peak-model} for any $P>0$.
A constant $f$ needs only zero encoders.
Independence of $1,\phi_{1,k},\phi_{2,k}$ excludes representation
\eqref{eq:affine} of these real encoders with $u,v$ independent over $\R$, so the universal characterization fails.

\subsubsection{Counterexample with Smooth Postprocessing Functions}

For the counterexample in Remark~\ref{rem:rigidity-smoothness},
let $K=4$, $\phi_{1,k}(s)=s$, $\phi_{2,k}(s)=s^2$, and
$$
 \rho(t)=
 \begin{cases}e^{-1/t},&t>0,\\0,&t\le0,\end{cases}
 \qquad
 \psi_1(u)=\rho(u-3/4),\qquad \psi_2(v)=v.
$$
Both receiver postprocessing functions $\psi_1$ and $\psi_2$ are smooth, but $\psi_1$ is constant
for $u\le3/4$. The decoder therefore needs the value of $m_1$ only
when $m_1>3/4$.
Set $q(s)=s-\rho(-s)$, $\alpha=\sqrt{P/5}$, and use the shared encoder
$\enc(s)=\alpha[s^2+iq(s)]$.
Since $q(s)\le s$, with equality for $s\ge0$, and $|q(s)|\le2$,
we have $|\enc(s)|^2=\alpha^2(s^4+q(s)^2)\le P$.
Write ${q_{\mathrm{av}}}=4^{-1}\sum_{k=1}^4q(s_k)$.

If any $s_k\le0$, then $m_1\le(0+1+1+1)/4=3/4$.
Thus, $m_1>3/4$ forces every input to be positive and gives
$q_{\mathrm{av}}=m_1$. Otherwise, $q_{\mathrm{av}}\le m_1\le3/4$,
so $\rho(q_{\mathrm{av}}-3/4)=\rho(m_1-3/4)=0$.
 Thus, for a noiseless received sample $y\in\mathcal R_\enc$, the decoder
$$
 \dec_0(y)=\frac{\Re y}{4\alpha}
 +\rho\!\left(\frac{\Im y}{4\alpha}-\frac34\right)
$$
recovers $f(\boldsymbol s)=m_2+\rho(m_1-3/4)$ exactly. The received components preserve $m_2$ and $q_{\mathrm{av}}$, which suffice for function computation even when $q_{\mathrm{av}}\ne m_1$. If $q$ belonged to the span of $1,s,s^2$, its agreement with $s$
on $(0,1)$ would force agreement everywhere, contradicting
$q(s)<s$ for $s<0$. Hence this encoder lies outside
\eqref{eq:affine}, although all prescribed functions are smooth.

\subsubsection{Scope of the Encoder Characterization}

When $1,\phi_{1,k},\phi_{2,k}$ are linearly independent
at every node and identifiable computation requires
\eqref{eq:affine} with $u,v$ linearly independent over $\R$,
continuous real-valued encoders $\mathscr E_k$ cannot permit exact computation of $f$.

Indeed, a real-valued encoder in \eqref{eq:affine} would satisfy
$\Im c_k+(\Im u)\phi_{1,k}(s)+(\Im v)\phi_{2,k}(s)=0$
for every $s$. Independence of the local triples forces
$\Im u=\Im v=0$, contradicting real linear independence of $u,v$.
 Dimension requirements for shared additive representations of symmetric functions are studied in~\cite{zaheer2017,wagstaff19,wagstaff22}. The fourth-node comparison requires $K\ge4$. For $K=3$, the affine form \eqref{eq:affine} remains sufficient, but we do not establish its necessity.

Identifiability alone does not ensure positive robustness.
For $K\ge4$, consider $\phi_{1,k}(s)=s$,
$\phi_{2,k}(s)=s^2$, and
$f(m_1,m_2)=\sqrt{1-m_1}+m_2$.
Along the attainable curve $(m_1,m_2)=(1-h,(1-h)^2)$,
the function-value change is of order $\sqrt h$, whereas
the moment displacement is of order $h$.
Thus, $f$ is not Lipschitz on $\mathcal M_K$.
Theorems~\ref{thm:rigidity} and~\ref{thm:sdp} therefore give
$\Lambda^\star(P;f)=0$, although encoders of the form
\eqref{eq:affine} with linearly independent $u,v$ permit
exact noiseless computation.

% \section{Proof of Optimal Design}\label{app:design-proof}

\subsection{Proof of Theorem~\ref{thm:sdp}}\label{app:general-robust}

\emph{Step 1: Every encoder in $\mathcal B_P^{\mathrm{all}}$ yields a feasible matrix solution.}
Fix $\enc\in\mathcal B_P^{\mathrm{all}}$ and inputs
$\boldsymbol s,\boldsymbol s'\in\mathbb I^K$ inducing moment vectors
$\boldsymbol m,\boldsymbol m'$, and write
$\boldsymbol d=\boldsymbol m-\boldsymbol m'$. 
In the noiseless aggregate difference, the node offsets cancel, giving
$$
|r_\enc(\boldsymbol s)-r_\enc(\boldsymbol s')|^2
=\|K\mathsf U\boldsymbol d\|_2^2
=K^2\boldsymbol d^\T G\boldsymbol d.
$$
Consequently, the pairwise separation constraint \eqref{eq:general-secant-program(c)} holds if and only if $\Lambda(\enc;f)\ge K^2\lambda$.
Since the real coefficient matrix $[\boldsymbol c_k\ \mathsf U]\in\R^{2\times 3}$ has two rows corresponding to the in-phase and quadrature components, its Gram matrix
$H_k=[\boldsymbol c_k\ \mathsf U]^\T[\boldsymbol c_k\ \mathsf U]$
satisfies $H_k\succeq0$ and $\operatorname{rank}(H_k)\le2$ and shares the common lower-right block $G=\mathsf U^\T\mathsf U$ in \eqref{eq:general-secant-program(a)}.
With $\boldsymbol q_k(s)=(1,\phi_{1,k}(s),\phi_{2,k}(s))^\T$, the peak-power constraint \eqref{eq:peak-model} becomes
$\boldsymbol q_k(s)^\T H_k\boldsymbol q_k(s)=|\enc_k(s)|^2\le P$,
which is \eqref{eq:general-secant-program(b)}.
Thus, every encoder $\enc\in\mathcal B_P^{\mathrm{all}}$ induces a feasible solution to \eqref{eq:general-secant-program} with $\lambda=\Lambda(\enc;f)/K^2$, establishing $\sup_{\enc\in\mathcal B_P^{\mathrm{all}}}\Lambda(\enc;f)\le K^2\lambda_\star$.

\emph{Step 2: Rank-two reconstruction and zero relaxation gap.}
Conversely, let $(H_1,\ldots,H_K,\lambda)$ be any feasible solution to \eqref{eq:general-secant-program}, where $H_k$ may have rank three. We construct a valid encoder in $\mathcal B_P^{\mathrm{all}}$ that preserves $G$ and $\lambda$ without increasing transmitted power.
Partition each matrix as $H_k=\left(\begin{smallmatrix}h_k&\boldsymbol\ell_k^\T\\
\boldsymbol\ell_k&G\end{smallmatrix}\right)\succeq0$.
We first establish the generalized Schur complement conditions $GG^\dagger\boldsymbol\ell_k=\boldsymbol\ell_k$ and $h_k\ge\boldsymbol\ell_k^\T G^\dagger\boldsymbol\ell_k$, where $G^\dagger$ inverts the positive eigenvalues of $G$ and leaves its zero eigenvalues unchanged:
\begin{itemize}
\item For any $\boldsymbol z\in\ker G$, evaluating the quadratic form of $H_k\succeq0$ at $(t,\boldsymbol z^\T)^\T$ gives $0\le h_kt^2+2t\boldsymbol\ell_k^\T\boldsymbol z$ for all $t\in\R$. Letting $t\to0$ with either sign forces $\boldsymbol\ell_k^\T\boldsymbol z=0$, so $\boldsymbol\ell_k\perp\ker G$, which implies $\boldsymbol\ell_k\in\operatorname{range}G$ and $GG^\dagger\boldsymbol\ell_k=\boldsymbol\ell_k$.
\item Evaluating the quadratic form of $H_k\succeq0$ at $(1,-(G^\dagger\boldsymbol\ell_k)^\T)^\T$ and using $GG^\dagger\boldsymbol\ell_k=\boldsymbol\ell_k$ gives $0\le h_k-2\boldsymbol\ell_k^\T G^\dagger\boldsymbol\ell_k+\boldsymbol\ell_k^\T G^\dagger GG^\dagger\boldsymbol\ell_k=h_k-\boldsymbol\ell_k^\T G^\dagger\boldsymbol\ell_k$.
\end{itemize}
Now define the encoder parameters $\mathsf U=G^{1/2}$ and $\boldsymbol c_k=G^{\dagger/2}\boldsymbol\ell_k$ (which remain well defined even when $G$ is singular). By construction, $\mathsf U^\T\mathsf U=G$, $\mathsf U^\T\boldsymbol c_k=GG^\dagger\boldsymbol\ell_k=\boldsymbol\ell_k$, and $\|\boldsymbol c_k\|_2^2=\boldsymbol\ell_k^\T G^\dagger\boldsymbol\ell_k\le h_k$.
The reconstructed rank-at-most-two Gram matrix
$$
\widetilde H_k:=[\boldsymbol c_k\ \mathsf U]^\T[\boldsymbol c_k\ \mathsf U]
=\begin{pmatrix}\|\boldsymbol c_k\|_2^2&\boldsymbol\ell_k^\T\\\boldsymbol\ell_k&G\end{pmatrix}
$$
differs from $H_k$ only in its $(1,1)$ entry. Since the first entry of $\boldsymbol q_k(s)$ is $1$, the transmitted power at node $k$ satisfies
$$
|\enc_k(s)|^2=\boldsymbol q_k(s)^\T\widetilde H_k\boldsymbol q_k(s)
=\boldsymbol q_k(s)^\T H_k\boldsymbol q_k(s)
-\underbrace{\bigl(h_k-\|\boldsymbol c_k\|_2^2\bigr)}_{\ge0}\le P.
$$
Because the lower-right block $G$ is unchanged, the squared pairwise separation distances $K^2\boldsymbol d^\T G\boldsymbol d$ and objective value $\lambda$ are preserved. Thus, omitting the rank constraint introduces no relaxation gap, and $\sup_{\enc\in\mathcal B_P^{\mathrm{all}}}\Lambda(\enc;f)=K^2\lambda_\star$.

\emph{Step 3: Uniform boundedness and attainment of $\lambda_\star$.}
To prove that $\lambda_\star$ is attained by an encoder in $\mathcal B_P^{\mathrm{all}}$, we take a sequence of feasible encoders with scores $K^2\lambda_n\to K^2\lambda_\star$ (the zero encoder is feasible at $\lambda=0$) and show that their coefficients can be chosen within a compact set.
Write $\boldsymbol a_k(s)=(\phi_{1,k}(s),\phi_{2,k}(s))^\T$, fix a reference input $s_0\in\mathbb I$, and let $\mathsf Q$ be the orthogonal projector onto the subspace of local preprocessing variations:
$$
\mathcal V=\operatorname{span}\{
\boldsymbol a_k(s)-\boldsymbol a_k(t):
s,t\in\mathbb I,\ k=1,\ldots,K\}\subseteq\R^2.
$$
Because $\boldsymbol a_k(s)-\boldsymbol a_k(s_0)\in\mathcal V$, the orthogonal component $(I_2-\mathsf Q)\boldsymbol a_k(s)=(I_2-\mathsf Q)\boldsymbol a_k(s_0)$ is independent of $s$. We can therefore absorb this static component into the node offset by setting $\mathsf U'=\mathsf U\mathsf Q$ and $\boldsymbol c_k'=\boldsymbol c_k+\mathsf U(I_2-\mathsf Q)\boldsymbol a_k(s_0)$. Indeed, for every $s\in\mathbb I$,
$$
\boldsymbol c_k'+\mathsf U'\boldsymbol a_k(s)
=\boldsymbol c_k+\mathsf U\boldsymbol a_k(s)
-\mathsf U(I_2-\mathsf Q)[\boldsymbol a_k(s)-\boldsymbol a_k(s_0)]
=\boldsymbol c_k+\mathsf U\boldsymbol a_k(s).
$$
Since every transmitted symbol is unchanged, peak feasibility \eqref{eq:general-secant-program(b)} and all separation constraints \eqref{eq:general-secant-program(c)} continue to hold at the same $\lambda$ with the new common Gram block $\mathsf U'^\T\mathsf U'=\mathsf QG\mathsf Q$.

Projecting onto $\mathcal V$ ensures that the remaining coefficients are uniformly bounded:
\begin{itemize}
\item \emph{Bounding $\mathsf U'$:} Let $d=\dim\mathcal V\le2$. If $d=0$, then $\mathsf U'=0$. If $d>0$, choose a fixed basis $D=[\boldsymbol d_1\ \cdots\ \boldsymbol d_d]$ of $\mathcal V$, where each $\boldsymbol d_i=\boldsymbol a_{k_i}(s_i)-\boldsymbol a_{k_i}(t_i)$ is a difference of local preprocessing vectors at some node $k_i$. Then, $\mathsf U'\boldsymbol d_i=\mathsf U\boldsymbol d_i=\boldsymbol\enc_{k_i}(s_i)-\boldsymbol\enc_{k_i}(t_i)$ is a difference of two transmitted symbols, so the peak constraint \eqref{eq:peak-model} implies $\|\mathsf U'\boldsymbol d_i\|_2\le2\sqrt P$. Using the left inverse $D^\dagger=(D^\T D)^{-1}D^\T$, which satisfies $DD^\dagger=\mathsf Q$, we obtain the uniform bound
$$
\|\mathsf U'\|_2=\|\mathsf U'DD^\dagger\|_2
\le\|\mathsf U'D\|_{\mathrm F}\|D^\dagger\|_2
\le2\sqrt{dP}\,\|D^\dagger\|_2.
$$
\item \emph{Bounding $\boldsymbol c_k'$ and $\lambda$:} Evaluating the peak constraint at $s_0$ gives $\|\boldsymbol c_k'\|_2\le\sqrt P+\|\mathsf U'\|_2\|\boldsymbol a_k(s_0)\|_2$. Moreover, since $f$ is nonconstant on $\mathcal M_K$, we can fix attainable $\boldsymbol m,\boldsymbol m'$ with $f(\boldsymbol m)\ne f(\boldsymbol m')$. The distance between their noiseless aggregates is at most $2K\sqrt P$, so \eqref{eq:general-secant-program(c)} yields $\lambda|f(\boldsymbol m)-f(\boldsymbol m')|^2\le\boldsymbol d^\T G\boldsymbol d\le4P$, bounding $\lambda$.
\end{itemize}
Applying this reconstruction and projection to a maximizing sequence with $\lambda_n\to\lambda_\star$ produces a uniformly bounded sequence $(\boldsymbol c_{1,n}',\ldots,\boldsymbol c_{K,n}',\mathsf U_n')$. By compactness, a subsequence converges to a limit $(\boldsymbol c_{1,\star},\ldots,\boldsymbol c_{K,\star},\mathsf U_\star)$. Passing to the limit in the pointwise inequalities for each $s\in\mathbb I$ and $\boldsymbol m,\boldsymbol m'\in\mathcal M_K$ gives
$$
\|\boldsymbol c_{k,\star}+\mathsf U_\star\boldsymbol a_k(s)\|_2^2\le P,\qquad
\|\mathsf U_\star(\boldsymbol m-\boldsymbol m')\|_2^2
\ge\lambda_\star|f(\boldsymbol m)-f(\boldsymbol m')|^2.
$$
Thus, the limiting encoder belongs to $\mathcal B_P^{\mathrm{all}}$ and attains the optimal score $K^2\lambda_\star$.

\emph{Step 4: Lipschitz equivalence and global optimality over $\mathcal E_P$.}
Suppose $f$ is Lipschitz on $\mathcal M_K$ with constant $M>0$. Because $\phi_{1,k},\phi_{2,k}$ are continuous on the compact interval $\mathbb I$, we can choose $\mathsf U=\alpha I_2$ and $\boldsymbol c_k=0$ with $\alpha>0$ small enough that $\alpha^2\max_{k,s}\|\boldsymbol a_k(s)\|_2^2\le P$. Then $|f(\boldsymbol m)-f(\boldsymbol m')|\le M\|\boldsymbol d\|_2$ implies $\Lambda(\enc;f)\ge K^2\alpha^2/M^2>0$, so $\lambda_\star>0$.
Conversely, if $\lambda_\star>0$, then \eqref{eq:general-secant-program(c)} at an optimal $G$ gives $\lambda_\star|f(\boldsymbol m)-f(\boldsymbol m')|^2\le\boldsymbol d^\T G\boldsymbol d\le\|G\|_2\|\boldsymbol d\|_2^2$, showing that $f$ is Lipschitz on $\mathcal M_K$ with constant $\sqrt{\|G\|_2/\lambda_\star}$.

Finally, because $\Lambda(\enc;f)>0$ implies identifiability in \eqref{eq:exactness}, under the assumptions of Theorem~\ref{thm:rigidity} every continuous peak-feasible encoder with a positive robustness score must take the form \eqref{eq:affine} and hence belong to $\mathcal B_P^{\mathrm{all}}$. Every encoder in $\mathcal E_P\setminus\mathcal B_P^{\mathrm{all}}$ has zero score, proving $\Lambda^\star(P;f)=K^2\lambda_\star$.

\subsection{Proof of the Gradient Bounds}\label{app:gradient-bounds}
\begin{proof}

Since $G\succ0$, the score in \eqref{eq:score} gives
$$
\frac{K^2}{\Lambda(\enc;f)}=
\sup_{\substack{\boldsymbol m,\boldsymbol m'\in\mathcal M_K\\
                 \boldsymbol m\ne\boldsymbol m'}}
\frac{|f(\boldsymbol m)-f(\boldsymbol m')|^2}
{(\boldsymbol m-\boldsymbol m')^\T G(\boldsymbol m-\boldsymbol m')}.
$$
For the upper bound in \eqref{eq:general-gradient-sandwich}, set
$\boldsymbol d=\boldsymbol m-\boldsymbol m'$. The segment
$\boldsymbol m'+t\boldsymbol d$, $0\le t\le1$, lies in $\mathcal C$.
The fundamental theorem of calculus and Cauchy--Schwarz give
\begin{align*}
|f(\boldsymbol m)-f(\boldsymbol m')|
&=\left|\int_0^1\boldsymbol w(\boldsymbol m'+t\boldsymbol d)^\T
\boldsymbol d\,dt\right|\\
&\le\int_0^1\|G^{-1/2}\boldsymbol w(\boldsymbol m'+t\boldsymbol d)\|_2
\|G^{1/2}\boldsymbol d\|_2\,dt\le\sqrt{B_{\mathcal C}(G)}\,\|\boldsymbol d\|_G.
\end{align*}
Squaring and taking the preceding supremum proves the upper bound.

For the lower bound, fix $\boldsymbol m\in\operatorname{int}\mathcal M_K$
with $\boldsymbol w(\boldsymbol m)\ne0$, and put
$\boldsymbol d=G^{-1}\boldsymbol w(\boldsymbol m)$.
The perturbed moment vector $\boldsymbol m+\varepsilon\boldsymbol d$ is
attainable for sufficiently small $\varepsilon$. Hence
$$
\frac{K^2}{\Lambda(\enc;f)}
\ge\lim_{\varepsilon\to0}
\frac{|f(\boldsymbol m+\varepsilon\boldsymbol d)-f(\boldsymbol m)|^2}
{\varepsilon^2\boldsymbol d^\T G\boldsymbol d}
=\frac{(\boldsymbol w(\boldsymbol m)^\T\boldsymbol d)^2}
{\boldsymbol d^\T G\boldsymbol d}
=\boldsymbol w(\boldsymbol m)^\T G^{-1}\boldsymbol w(\boldsymbol m).
$$
A zero gradient contributes zero. Taking the supremum over interior
moment vectors proves the lower bound. If $\mathcal M_K$ is compact, convex,
and full-dimensional, then
$\mathcal M_K=\overline{\operatorname{int}\mathcal M_K}$.
Continuity of $\boldsymbol w$ makes the interior supremum equal to
$B_{\mathcal M_K}(G)$, so choosing $\mathcal C=\mathcal M_K$
makes both bounds equal.

\end{proof}

\paragraph{Failure of the two-endpoint reduction}

The endpoint reduction in Lemma~\ref{prop:score} does not
extend to arbitrary receiver postprocessing functions,
as the following example shows. Take $K=4$, $P=2$, $\phi_1(s)=s$, $\phi_2(s)=s^2$, and
$\enc(s)=s+is^2$, so $|\enc(s)|^2\le2$ and $G=I_2$.
Set $\psi_1(m_1)=m_1-m_1^3/3$ and $\psi_2(m_2)=m_2$.
Then $\boldsymbol w=(1-m_1^2,1)^\T$, giving
$\|\boldsymbol w\|_2^2\le2$ on $|m_1|\le1$, with equality at
the interior attainable moment vector $(0,1/2)^\T$. Applying \eqref{eq:general-gradient-sandwich} with
$\mathcal C=[-1,1]\times\R$ gives
$K^2/\Lambda(\enc;f)=2$, hence $\Lambda(\enc;f)=8$. The endpoint gradients $(0,1)^\T$ would instead predict $16$.

\section{Decoding Proofs}\label{app:decoding-proofs}
\subsection{Proof of Proposition~\ref{prop:amplification}}\label{app:amplification-proof}

Setting $z=0$ in \eqref{eq:uniform} forces exact noiseless recovery, $\dec(r_\enc(\boldsymbol s))=f(\boldsymbol s)$.
For any $\boldsymbol s,\boldsymbol s'\in\mathbb I^K$, choosing the noise realization $z=r_\enc(\boldsymbol s')-r_\enc(\boldsymbol s)$ shifts $r_\enc(\boldsymbol s)$ to $r_\enc(\boldsymbol s')$, so \eqref{eq:uniform} yields
$$
|f(\boldsymbol s')-f(\boldsymbol s)|
=|\dec(r_\enc(\boldsymbol s)+z)-f(\boldsymbol s)|
\le L|z|=L|r_\enc(\boldsymbol s')-r_\enc(\boldsymbol s)|.
$$
Dividing by $|f(\boldsymbol s')-f(\boldsymbol s)|>0$ and taking the infimum in \eqref{eq:score} gives $L\ge\Lambda(\enc;f)^{-1/2}$ (with $L_\enc=+\infty$ when $\Lambda(\enc;f)=0$). When $\Lambda(\enc;f)>0$, set $L=\Lambda(\enc;f)^{-1/2}$. By \eqref{eq:score}, the noiseless decoder $\dec_0(r_\enc(\boldsymbol s)):=f(\boldsymbol s)$ is well defined and $L$-Lipschitz on $\mathcal R_\enc$, obeying $|\dec_0(r)-\dec_0(r')|\le L|r-r'|$.
McShane's theorem~\cite{mcshane1934} extends $\dec_0$ to a continuous, $L$-Lipschitz map on all of $\C$ via
$$
\dec(y):=\inf_{r'\in\mathcal R_\enc}
\{\dec_0(r')+L|y-r'|\}.
$$
Because $\dec$ agrees with $\dec_0$ on $\mathcal R_\enc$, we have $\dec(r_\enc(\boldsymbol s))=f(\boldsymbol s)$, and its $L$-Lipschitz property gives
$|\dec(r_\enc(\boldsymbol s)+z)-f(\boldsymbol s)|=|\dec(r_\enc(\boldsymbol s)+z)-\dec(r_\enc(\boldsymbol s))|\le L|z|$, attaining $L_\enc=L$.

\subsection{Proof of Theorem~\ref{thm:decoder}}\label{app:general-decoder-proof}

Under \eqref{eq:affine}, $\mathcal R_\enc$ consists of all points $\boldsymbol r'=K\boldsymbol c+K\mathsf U\boldsymbol m$ for $\boldsymbol m\in\mathcal M_K$, with $\dec_0(r')=f(\boldsymbol m)$. 
Equation~\eqref{eq:general-extension-decoder} is therefore McShane's extension from Appendix~\ref{app:amplification-proof}; compactness of $\mathcal M_K$ ensures that the minimum is attained even for singular $\mathsf U$, and Proposition~\ref{prop:amplification} gives exact noiseless recovery with optimal constant $L_\enc=\Lambda(\enc;f)^{-1/2}$. 
When $\mathbb E|z|^2=\sigma_z^2$, \eqref{eq:uniform} directly yields $\mathbb E_z[|\dec_\enc(r_\enc(\boldsymbol s)+z)-f(\boldsymbol s)|^2]\le L_\enc^2\mathbb E|z|^2=\sigma_z^2/\Lambda(\enc;f)$. 
Clipping the output to a closed interval $I_f\supseteq f(\mathbb I^K)$ cannot increase its distance to $f(\boldsymbol s)\in I_f$, preserving all bounds.

\paragraph{Convex-projection decoder}
For $\tilde{\boldsymbol m}=\boldsymbol m+K^{-1}\mathsf U^{-1}\boldsymbol z$, let $\hat{\boldsymbol m}=\Pi_{\mathcal C}^G\tilde{\boldsymbol m}$. 
Since $\boldsymbol m\in\mathcal M_K\subseteq\mathcal C$, projection optimality gives $(\tilde{\boldsymbol m}-\hat{\boldsymbol m})^\T G(\boldsymbol m-\hat{\boldsymbol m})\le0$, so
$$
\|\hat{\boldsymbol m}-\boldsymbol m\|_G^2
\le(\hat{\boldsymbol m}-\boldsymbol m)^\T G(\tilde{\boldsymbol m}-\boldsymbol m)
\le\|\hat{\boldsymbol m}-\boldsymbol m\|_G\|\tilde{\boldsymbol m}-\boldsymbol m\|_G.
$$
Thus $\|\hat{\boldsymbol m}-\boldsymbol m\|_G\le\|\tilde{\boldsymbol m}-\boldsymbol m\|_G=\|K^{-1}\mathsf U^{-1}\boldsymbol z\|_G=|z|/K$, where the last equality uses $G=\mathsf U^\T\mathsf U$. 
The gradient bound in Appendix~\ref{app:gradient-bounds} then yields
$$
|f(\hat{\boldsymbol m})-f(\boldsymbol m)|
\le\sqrt{B_{\mathcal C}(G)}\,\|\hat{\boldsymbol m}-\boldsymbol m\|_G
\le\frac{\sqrt{B_{\mathcal C}(G)}}K|z|,
$$
which attains $L_\enc$ when $B_{\mathcal C}(G)=K^2/\Lambda(\enc;f)$.

\subsection{Proof of Proposition~\ref{prop:radius-minimax}}\label{app:radius-minimax-proof}

Choose inputs $\boldsymbol s,\boldsymbol s'$ whose noiseless aggregates are at distance at most $2r_n$, and let $y_0=[r_\enc(\boldsymbol s)+r_\enc(\boldsymbol s')]/2$. 
Both inputs can produce $y_0$ under a noise realization of magnitude at most $r_n$. 
The triangle inequality implies that for any decoder $\dec$,
$$
\max\bigl\{|\dec(y_0)-f(\boldsymbol s)|,\ |\dec(y_0)-f(\boldsymbol s')|\bigr\}
\ge\tfrac12|f(\boldsymbol s)-f(\boldsymbol s')|.
$$
Maximizing over all such input pairs establishes the lower bound $\omega_\enc(2r_n)/2$; continuity and compactness ensure that this supremum is attained.

For any possible observation
$y\in\mathcal R_\enc+\{z\in\C:|z|\le r_n\}$, the set of consistent inputs $\{\boldsymbol s\in\mathbb I^K:|y-r_\enc(\boldsymbol s)|\le r_n\}$ is nonempty and compact. 
Let $f_-(y)$ and $f_+(y)$ denote its minimum and maximum desired function values, respectively. 
For any two consistent inputs $\boldsymbol s$ and $\boldsymbol s'$, the triangle inequality bounds their aggregate separation by $|r_\enc(\boldsymbol s)-r_\enc(\boldsymbol s')|\le|r_\enc(\boldsymbol s)-y|+|y-r_\enc(\boldsymbol s')|\le2r_n$. 
Consequently, the range of consistent function values has width $f_+(y)-f_-(y)\le\omega_\enc(2r_n)$. 
The midpoint decoder outputs $\bigl[f_-(y)+f_+(y)\bigr]/2$, yielding an absolute error for any consistent input $\boldsymbol s$ of at most
$$
\left|\frac{f_-(y)+f_+(y)}2-f(\boldsymbol s)\right|
\le\frac{f_+(y)-f_-(y)}2
\le\frac{\omega_\enc(2r_n)}2,
$$
thereby achieving the optimal worst-case error bound. To check measurability, note that for each $a\in\R$, the level sets $\{y:f_+(y)\ge a\}$ and
$\{y:f_-(y)\le a\}$ within the possible-observation set are
continuous images of compact input--noise sets and hence closed.
Thus, both endpoint functions and their midpoint are Borel measurable.
Extending the decoder by zero outside the compact possible-observation
set preserves measurability. This also covers $r_n=0$.

% To check measurability, take $y_j\to y$ in the possible-observation set.
% Choose maximizing inputs along a subsequence realizing
% $\limsup_j f_+(y_j)$. Compactness gives a further convergent
% subsequence; continuity makes its limit consistent with $y$.
% Thus $\limsup_j f_+(y_j)\le f_+(y)$, and similarly
% $\liminf_j f_-(y_j)\ge f_-(y)$.
% Both endpoint functions, and hence their midpoint, are Borel measurable.
% The possible-observation set is compact, so extending the decoder
% by zero elsewhere preserves measurability. This also covers $r_n=0$.

\section{Proofs of the Gaussian Estimation Limits}\label{app:statistical-converse}

 We use a weighted van Trees argument~\cite{gill1995,tsang2020} for the converse, which requires no derivatives of the decoder $\dec$.

\subsection{Proof of Theorem~\ref{thm:identifiable-minimax}}\label{app:identifiable-minimax-proof}

\emph{Step 1: Fixed-encoder achievability.}
Fix $\enc\in\mathcal B_P$. Since $\mathcal M_K=\overline{\operatorname{int}\mathcal M_K}$, both attainable coordinate intervals $\mathbb J_1,\mathbb J_2$ have positive length.
Extend each receiver postprocessing function $\psi_j$ linearly beyond $\mathbb J_j$ using its one-sided endpoint derivatives, so that $\bar\psi_j\in C^1(\R)$ has a bounded, uniformly continuous derivative. 
This defines a noise-power-independent decoder $\dec$ even when noise moves the unconstrained moment estimate $\tilde{\boldsymbol m}$ outside $\mathbb J_1\times\mathbb J_2$:
$$
 \dec(y)=\bar\psi_1(\tilde m_1)+\bar\psi_2(\tilde m_2),
 \qquad
 \tilde{\boldsymbol m}
 =K^{-1}\mathsf U^{-1}(\boldsymbol y-K\boldsymbol c).
$$
Write $z=\sigma_z z_0$ with $z_0\sim\mathcal{CN}(0,1)$, let $\boldsymbol z_0=(\Re z_0,\Im z_0)^\T$, and set $\boldsymbol\delta=\sigma_z\mathsf U^{-1}\boldsymbol z_0/K$.
For $\bar f(\boldsymbol m):=\bar\psi_1(m_1)+\bar\psi_2(m_2)$, uniform continuity of the gradient on $\mathcal M_K$ gives the first-order expansion
\begin{align}
 \bar f(\boldsymbol m+\boldsymbol\delta)-f(\boldsymbol m)
 &=\boldsymbol w(\boldsymbol m)^\T\boldsymbol\delta
   +r(\boldsymbol m,\boldsymbol\delta),\\
 \sup_{\boldsymbol m\in\mathcal M_K}
 |r(\boldsymbol m,\boldsymbol\delta)|
 &\le\|\boldsymbol\delta\|_2\,\omega(\|\boldsymbol\delta\|_2),
\end{align}
where $\omega$ is bounded and $\omega(t)\to0$ as $t\downarrow0$.
Since $\|\boldsymbol\delta\|_2\le\sigma_z\|\mathsf U^{-1}\|_2|z_0|/K$, dominated convergence implies that the normalized remainder vanishes uniformly:
$$
 \frac{\mathbb E\sup_{\boldsymbol m\in\mathcal M_K}
 |r(\boldsymbol m,\boldsymbol\delta)|^2}{\sigma_z^2}
 \le\frac{\|\mathsf U^{-1}\|_2^2}{K^2}
 \mathbb E\!\left[|z_0|^2\omega(\|\boldsymbol\delta\|_2)^2\right]
 \longrightarrow0.
$$
Because $\boldsymbol w$ is bounded on $\mathcal M_K$, Cauchy--Schwarz shows that the cross term is $o(\sigma_z^2)$ uniformly in $\boldsymbol m$:
$$
 \sup_{\boldsymbol m\in\mathcal M_K}\left|\mathbb E[
 (\boldsymbol w(\boldsymbol m)^\T\boldsymbol\delta)
 r(\boldsymbol m,\boldsymbol\delta)]\right|
 \le O(\sigma_z)\left(\mathbb E\sup_{\boldsymbol m\in\mathcal M_K}
 |r(\boldsymbol m,\boldsymbol\delta)|^2\right)^{1/2}
 =o(\sigma_z^2).
$$
Evaluating the leading quadratic term with $\mathbb E[\boldsymbol z_0\boldsymbol z_0^\T]=I_2/2$ and $\mathsf U^{-1}\mathsf U^{-\T}=G^{-1}$ yields
\begin{equation}
 \sup_{\boldsymbol m\in\mathcal M_K}
 \left|
 \frac{\mathbb E_z
 |\dec(Kc+K(um_1+vm_2)+z)
   -f(\boldsymbol m)|^2}{\sigma_z^2}
 -\frac{\boldsymbol w(\boldsymbol m)^\T
 G^{-1}\boldsymbol w(\boldsymbol m)}{2K^2}
 \right|\longrightarrow0.
\label{eq:uniform-gaussian-expansion}
\end{equation}
Taking the supremum over $\boldsymbol m\in\mathcal M_K$ establishes the achievable upper bound in \eqref{eq:fixed-minimax-limit}.

\emph{Step 2: Fixed-encoder and sequential converse.}
Let $(\enc_n,\dec_n)$ be a sequence of encoder--decoder pairs with Gram matrices $G_n\to G_0\succ0$ as $\sigma_{z,n}^2\to0$.
For any interior moment vector $\boldsymbol m_0\in\operatorname{int}\mathcal M_K$, let $\nu_r$ be the uniform distribution on an open ball $B(\boldsymbol m_0,r)\subset\operatorname{int}\mathcal M_K$. 
Since worst-case risk is at least the average risk ($R_{\max}\ge R_{\nu_r}$), the sequential average-MSE lower bound \eqref{genstat:changing-limit} proved in Appendix~\ref{genstat:information} below gives
$\liminf_n R_{\max}(\enc_n,\dec_n;\sigma_{z,n}^2)/\sigma_{z,n}^2\ge\mathcal J_{\nu_r}(G_0)$.
Letting $r\downarrow0$, continuity of $\boldsymbol w$ gives $\mathcal J_{\nu_r}(G_0)\to\boldsymbol w(\boldsymbol m_0)^\T G_0^{-1}\boldsymbol w(\boldsymbol m_0)/(2K^2)$.
Since $\mathcal M_K=\overline{\operatorname{int}\mathcal M_K}$, taking the supremum over $\boldsymbol m_0\in\operatorname{int}\mathcal M_K$ yields
\begin{equation}
 \liminf_n\frac{R_{\max}(\enc_n,\dec_n;\sigma_{z,n}^2)}{\sigma_{z,n}^2}
 \ge\frac{1}{2K^2}\max_{\boldsymbol m\in\mathcal M_K}
 \boldsymbol w(\boldsymbol m)^\T G_0^{-1}\boldsymbol w(\boldsymbol m).
 \label{eq:localized-minimax-converse}
\end{equation}
For a fixed encoder ($G_n=G$), choosing decoders $\dec_n$ within $\sigma_{z,n}^4$ of $\inf_\dec R_{\max}(\enc,\dec;\sigma_{z,n}^2)$ along a sequence realizing the lower limit proves \eqref{eq:fixed-minimax-limit}.

\emph{Step 3: Optimization over $\mathcal B_P$ and attainment.}
Now suppose $1,\phi_{1,k},\phi_{2,k}$ are linearly independent for each $k$ and $\{\boldsymbol w(\boldsymbol m):\boldsymbol m\in\operatorname{int}\mathcal M_K\}$ spans $\R^2$.
Choose $\boldsymbol m_a,\boldsymbol m_b\in\operatorname{int}\mathcal M_K$ with linearly independent gradients $\boldsymbol w(\boldsymbol m_a),\boldsymbol w(\boldsymbol m_b)$, and let $\nu$ be a regular moment distribution supported on the union of small interior balls around $\boldsymbol m_a$ and $\boldsymbol m_b$.
If $\boldsymbol v^\T Q_\nu\boldsymbol v=\mathbb E_\nu[(\boldsymbol v^\T\boldsymbol w)^2]=0$, continuity forces $\boldsymbol v^\T\boldsymbol w(\boldsymbol m_a)=\boldsymbol v^\T\boldsymbol w(\boldsymbol m_b)=0$ and hence $\boldsymbol v=0$, so $Q_\nu\succ0$.

To establish compactness of the coefficient set, we use the fact that, for each node $k$, linear independence of $1,\phi_{1,k},\phi_{2,k}$ provides three inputs $s_{k,1},s_{k,2},s_{k,3}\in\mathbb I$ such that the $3\times 3$ matrix $Z_k$ with columns $(1,\phi_{1,k}(s_{k,r}),\phi_{2,k}(s_{k,r}))^\T$ is invertible.
For $A_k=[\boldsymbol c_k\ \mathsf U]$, the columns of $A_kZ_k$ are the transmitted symbols at these three inputs, so the peak constraint \eqref{eq:peak-model} bounds all coefficients uniformly:
$$
 \|A_k\|_{\mathrm F}
 \le\|A_kZ_k\|_{\mathrm F}\|Z_k^{-1}\|_2
 \le\sqrt{3P}\,\|Z_k^{-1}\|_2.
$$
Thus, the peak-feasible coefficient set is compact when singular $\mathsf U$ is allowed.
Moreover, for sufficiently small $\varepsilon>0$, the encoder with $\mathsf U=\varepsilon I_2$ and $c_k=0$ is peak-feasible and achieves a finite value in \eqref{eq:fixed-minimax-limit}, while any $G\succ0$ satisfies
$$
 \max_{\boldsymbol m\in\mathcal M_K}
 \boldsymbol w(\boldsymbol m)^\T G^{-1}\boldsymbol w(\boldsymbol m)
 \ge\operatorname{tr}(Q_\nu G^{-1})
 \ge\frac{\lambda_{\min}(Q_\nu)}{\lambda_{\min}(G)}.
$$
Because $\lambda_{\min}(Q_\nu)>0$, the objective diverges as $\lambda_{\min}(G)\downarrow0$, excluding singular limits.
Compactness and continuity therefore guarantee an encoder $\enc_\star\in\mathcal B_P$ with $G_\star\succ0$ that attains the right-hand side of \eqref{eq:optimized-minimax-limit}. Applying \eqref{eq:uniform-gaussian-expansion} to $\enc_\star$ proves the achievable upper bound.

Finally, for the optimized lower bound, choose pairs $(\enc_n,\dec_n)$ within $\sigma_{z,n}^4$ of the minimax risk infima along a sequence $\sigma_{z,n}^2\downarrow0$ realizing the lower limit in \eqref{eq:optimized-minimax-limit}.
Their normalized risks remain bounded by the achievable upper bound obtained with $\enc_\star$, up to a term that vanishes as $n\to\infty$.
By compactness of the coefficient set, a subsequence converges to a peak-feasible limit with $G_n\to G_0\succeq0$.
If $G_0$ were singular, the singular-limit argument in Appendix~\ref{genstat:information} applied to the auxiliary law $\nu$ (where $Q_\nu\succ0$) would force $R_{\max}/\sigma_{z,n}^2\ge R_\nu/\sigma_{z,n}^2\to\infty$, a contradiction.
Hence, $G_0\succ0$, the limiting encoder $\enc_0$ lies in $\mathcal B_P$, and \eqref{eq:localized-minimax-converse} lower-bounds the limit by the coefficient of $\enc_0$, completing the proof of \eqref{eq:optimized-minimax-limit}.

\subsection{Proof of Proposition~\ref{prop:general-small-noise}}\label{genstat:information}

\emph{Step 1: Achievability.}
Fix a regular moment distribution $\pi$ and an encoder $\enc\in\mathcal B_P$ (so $\mathsf U$ is invertible and $G\succ0$). 
Since $\psi_1',\psi_2'$ are continuous on the compact intervals $\mathbb J_1,\mathbb J_2$, $f$ is $M$-Lipschitz on $\mathcal M_K$ for some $M>0$, giving $\Lambda(\enc;f)\ge K^2\lambda_{\min}(G)/M^2>0$. 
Theorem~\ref{thm:decoder} therefore provides a globally Lipschitz decoder $\dec_\enc$ that is exact on $\mathcal R_\enc$. Any fixed $L$-Lipschitz decoder $\dec$ that is exact on $\mathcal R_\enc$ achieves $\mathcal J_\pi(G)$, as follows.

Let $y_{\boldsymbol m}:=Kc+K(um_1+vm_2)$ and write $z=\sigma_z z_0$ with $z_0\sim\mathcal{CN}(0,1)$ and $\boldsymbol z_0=(\Re z_0,\Im z_0)^\T$. 
In real coordinates, the noisy observation is $\boldsymbol y_{\boldsymbol m}+\sigma_z\boldsymbol z_0=K\boldsymbol c+K\mathsf U\bigl(\boldsymbol m+\frac{\sigma_z}{K}\mathsf U^{-1}\boldsymbol z_0\bigr)$. 
For each interior moment vector $\boldsymbol m\in\Omega\subset\operatorname{int}\mathcal M_K$ and fixed $z_0$, taking $\sigma_z>0$ sufficiently small keeps the perturbed moment vector inside $\mathcal M_K$:
$$
 \boldsymbol m+\frac{\sigma_z}{K}\mathsf U^{-1}\boldsymbol z_0
 \in\mathcal M_K
 \implies
 \dec(y_{\boldsymbol m}+\sigma_z z_0)
 =f\!\left(\boldsymbol m+
 \frac{\sigma_z}{K}\mathsf U^{-1}\boldsymbol z_0\right),
$$
where the equality follows from exact noiseless recovery on $\mathcal R_\enc$. Differentiating $f$ at $\boldsymbol m$ as $\sigma_z\downarrow0$ gives
$$
 \frac{\dec(y_{\boldsymbol m}+\sigma_z z_0)-f(\boldsymbol m)}
 {\sigma_z}
 \longrightarrow
 \frac{1}{K}\boldsymbol w(\boldsymbol m)^\T
 \mathsf U^{-1}\boldsymbol z_0.
$$
Because $\dec(y_{\boldsymbol m})=f(\boldsymbol m)$ and $\dec$ is $L$-Lipschitz, the squared quotient is dominated by $|\dec(y_{\boldsymbol m}+\sigma_z z_0)-\dec(y_{\boldsymbol m})|^2/\sigma_z^2\le L^2|z_0|^2$, with $\mathbb E|z_0|^2=1$. 
Since $\pi(\Omega)=1$, $\mathbb E[\boldsymbol z_0\boldsymbol z_0^\T]=I_2/2$, and $\mathsf U^{-1}\mathsf U^{-\T}=G^{-1}$, dominated convergence yields
$$
 \lim_{\sigma_z^2\downarrow0}
 \frac{R_\pi(\enc,\dec;\sigma_z^2)}{\sigma_z^2}
 =\frac{1}{2K^2}\mathbb E_\pi\!\left[\boldsymbol w(\boldsymbol m)^\T G^{-1}\boldsymbol w(\boldsymbol m)\right]
 =\frac{\operatorname{tr}(Q_\pi G^{-1})}{2K^2}
 =\mathcal J_\pi(G).
$$
In particular, $\dec_\enc$ achieves $\mathcal J_\pi(G)$.

\emph{Step 2: Non-asymptotic information bound.}
Let $p$ be the continuously differentiable density of $\pi$ on $\Omega\subset\operatorname{int}\mathcal M_K$.
A test function $\boldsymbol\eta\in C_c^1(\Omega;\R^2)$ selects weighted moment directions, and its compact support inside $\Omega$ eliminates boundary terms when integrating by parts. Define
$$
 A(\boldsymbol\eta)=\mathbb E_\pi[\boldsymbol\eta^\T\boldsymbol w],\quad
 B_G(\boldsymbol\eta)=\mathbb E_\pi[\boldsymbol\eta^\T G\boldsymbol\eta],\quad
 B_p(\boldsymbol\eta)=\int_\Omega
 \frac{[\operatorname{div}(p\boldsymbol\eta)]^2}{p}\,d\boldsymbol m,
$$
where $\operatorname{div}(p\boldsymbol\eta)=\partial(p\eta_1)/\partial m_1+\partial(p\eta_2)/\partial m_2$.
Here, $A(\boldsymbol\eta)$ measures average function sensitivity, $K^2B_G(\boldsymbol\eta)$ measures average squared received-signal sensitivity, and $B_p(\boldsymbol\eta)$ accounts for the prior density $p$.
We show that for any Borel-measurable decoder $\dec$ and $G\succ0$,
\begin{equation}
 \frac{R_\pi(\enc,\dec;\sigma_z^2)}{\sigma_z^2}
 \ge\frac{A(\boldsymbol\eta)^2}
 {2K^2 B_G(\boldsymbol\eta)+\sigma_z^2B_p(\boldsymbol\eta)}
 \label{genstat:information-bound}
\end{equation}
whenever the denominator is positive.
To prove \eqref{genstat:information-bound} without assuming decoder regularity, first clip $\dec$ to the bounded range of $f$, which can only reduce $R_\pi$ and ensures finite second moments.
Under the Gaussian observation model, define the weighted score of the joint moment--observation density:
$$
 S_{\boldsymbol\eta}=
 \frac{\operatorname{div}(p\boldsymbol\eta)}p
 +\frac{2K}{\sigma_z^2}\boldsymbol\eta^\T\mathsf U^\T
 (\boldsymbol y-K\boldsymbol c-K\mathsf U\boldsymbol m).
$$
Conditional on $\boldsymbol m$, the noise term in $S_{\boldsymbol\eta}$ has zero mean and variance $2K^2\boldsymbol\eta^\T G\boldsymbol\eta/\sigma_z^2$.
Because $p$ is bounded away from zero and $p,\nabla p$ are bounded on $\operatorname{supp}(\boldsymbol\eta)$, integration by parts with respect to $\boldsymbol m$ yields
$$
 \mathbb E[S_{\boldsymbol\eta}^2]
 =B_p(\boldsymbol\eta)+\frac{2K^2}{\sigma_z^2}B_G(\boldsymbol\eta),\qquad
 \mathbb E[(\dec(y)-f(\boldsymbol m))S_{\boldsymbol\eta}]
 =\mathbb E_\pi[\boldsymbol\eta^\T\boldsymbol w]
 =A(\boldsymbol\eta).
$$
Applying Cauchy--Schwarz gives $A(\boldsymbol\eta)^2\le R_\pi\mathbb E[S_{\boldsymbol\eta}^2]$, which rearranges into \eqref{genstat:information-bound}.

\emph{Step 3: Variational supremum and asymptotic converse.}
For any $G\succ0$, we next show that
\begin{equation}
 \sup_{\boldsymbol\eta\in C_c^1(\Omega;\R^2),\,\boldsymbol\eta\ne0}
 \frac{A(\boldsymbol\eta)^2}{B_G(\boldsymbol\eta)}
 =\operatorname{tr}(Q_\pi G^{-1}).
 \label{genstat:variational}
\end{equation}
Indeed, Cauchy--Schwarz applied to $\boldsymbol\eta^\T\boldsymbol w=(G^{1/2}\boldsymbol\eta)^\T(G^{-1/2}\boldsymbol w)$ gives
\begin{align*}
 A(\boldsymbol\eta)^2
 &\le B_G(\boldsymbol\eta)\mathbb E_\pi[\boldsymbol w^\T G^{-1}\boldsymbol w]\\
 &=B_G(\boldsymbol\eta)\operatorname{tr}(Q_\pi G^{-1}).
\end{align*}
Conversely, if $Q_\pi\ne0$, set $\boldsymbol\eta_*=G^{-1}\boldsymbol w$.
Since $\boldsymbol w$ is bounded and continuous, multiplying $\boldsymbol\eta_*$ by compactly supported interior cutoffs and smoothing via mollifiers produces a sequence $\boldsymbol\eta_j\in C_c^1(\Omega;\R^2)$ satisfying $B_G(\boldsymbol\eta_j-\boldsymbol\eta_*)\to0$.
Writing $T:=\mathbb E_\pi[\boldsymbol w^\T G^{-1}\boldsymbol w]>0$, we have $A(\boldsymbol\eta_j)\to T$ and $B_G(\boldsymbol\eta_j)\to T$, so $A(\boldsymbol\eta_j)^2/B_G(\boldsymbol\eta_j)\to T$, establishing \eqref{genstat:variational} without requiring differentiability of $\boldsymbol w$ or finite global prior Fisher information (when $Q_\pi=0$, both sides are zero). Now consider any sequence of encoder--decoder pairs $(\enc_n,\dec_n)$ with Gram matrices $G_n\to G_0\succeq0$ as $\sigma_{z,n}^2\to0$:
\begin{itemize}
\item \emph{Nonsingular limit ($G_0\succ0$):} For each fixed nonzero $\boldsymbol\eta\in C_c^1(\Omega;\R^2)$, we have $B_{G_n}(\boldsymbol\eta)\to B_{G_0}(\boldsymbol\eta)>0$ and $\sigma_{z,n}^2B_p(\boldsymbol\eta)\to0$. Taking $\liminf_n$ in \eqref{genstat:information-bound} and then taking the supremum over $\boldsymbol\eta$ via \eqref{genstat:variational} gives
\begin{equation}
 \liminf_n\frac{R_\pi(\enc_n,\dec_n;\sigma_{z,n}^2)}{\sigma_{z,n}^2}
 \ge\mathcal J_\pi(G_0).
 \label{genstat:changing-limit}
\end{equation}
\item \emph{Singular limit ($\det G_0=0$ with $Q_\pi\succ0$):} Take a nonzero vector $\boldsymbol v\in\ker G_0$. Since $\mathbb E_\pi[(\boldsymbol v^\T\boldsymbol w)^2]=\boldsymbol v^\T Q_\pi\boldsymbol v>0$, continuity of $\boldsymbol w$ ensures that $\boldsymbol v^\T\boldsymbol w$ has a constant nonzero sign on some open ball inside $\Omega$. Choosing a nonnegative, nonzero smooth function $\chi$ supported on that ball and setting $\boldsymbol\eta=\chi\boldsymbol v$ gives $A(\boldsymbol\eta)=\mathbb E_\pi[\chi\boldsymbol v^\T\boldsymbol w]\ne0$, whereas $B_{G_n}(\boldsymbol\eta)=\mathbb E_\pi[\chi^2]\boldsymbol v^\T G_n\boldsymbol v\to0$ and $\sigma_{z,n}^2B_p(\boldsymbol\eta)\to0$. By \eqref{genstat:information-bound}, $R_\pi(\enc_n,\dec_n;\sigma_{z,n}^2)/\sigma_{z,n}^2\to\infty$.
\end{itemize}
Finally, for a fixed encoder $\enc\in\mathcal B_P$ ($G_n=G\succ0$), choosing $\sigma_{z,n}^2\downarrow0$ attaining the lower limit in \eqref{eq:fixed-bayes-limit} and decoders $\dec_n$ within $\sigma_{z,n}^4$ of $\inf_\dec R_\pi(\enc,\dec;\sigma_{z,n}^2)$ gives $\liminf_{\sigma_z^2\downarrow0}\inf_\dec R_\pi(\enc,\dec;\sigma_z^2)/\sigma_z^2\ge\mathcal J_\pi(G)$ by \eqref{genstat:changing-limit}. Combining this with Step~1 completes the proof of Proposition~\ref{prop:general-small-noise}.

\subsection{Local Sensitivity and Global Robustness}\label{genstat:counterexamples}

Take $K=4$, $\phi_1(s)=s$, $\phi_2(s)=s^2$, and the shared encoder $\enc(s)=s+10is^2$. With $P=101$, this encoder is peak-feasible and $G=\operatorname{diag}(1,100)$.
Appendix~\ref{app:moment-region-proof} describes the attainable moment set $\mathcal M_4$, which is the closure of its interior. For $0\le m_1\le1/2$, the upper boundary is $m_{2,\max}(m_1)=1-2m_1+4m_1^2$.

Using the clipping operation defined after \eqref{eq:polynomial-decoder}, define the continuous functions
\begin{align}
h(x)&=\operatorname{clip}_{[0,1]}\bigl(40(x-1/10)\bigr)
      -\operatorname{clip}_{[0,1]}\bigl(40(x-3/8)\bigr),\\
r(x)&=\operatorname{clip}_{[0,1]}\bigl(100(x-21/25)\bigr).
\end{align}
The function $h$ is a unit-height trapezoid; $r$ is a ramp clipped at one.
Choose the receiver postprocessing functions $\psi_1(m_1)=\int_0^{m_1}h(t)\,dt$ and $\psi_2(m_2)=-10\int_0^{m_2}r(t)\,dt$.

Both postprocessing functions are $C^1$, with
$\boldsymbol w(\boldsymbol m)=(h(m_1),-10r(m_2))^\T$.
For attainable moment vectors,
$h(m_1)\ne0\Rightarrow 1/10<m_1<2/5
\Rightarrow m_2\le m_{2,\max}(m_1)<21/25
\Rightarrow r(m_2)=0$.
Since $0\le h,r\le1$ and $r(1)=1$ at the attainable moment vector $(0,1)$,
$$\sup_{\mathcal M_4}\boldsymbol w^\T G^{-1}\boldsymbol w
 =\sup_{\mathcal M_4}(h(m_1)^2+r(m_2)^2)=1.$$

Both $\boldsymbol m_L=(0,1)^\T$ and $\boldsymbol m_R=(1/2,9/10)^\T$ are attainable, using inputs $(-1,-1,1,1)$ and $(1,1,\sqrt{4/5},-\sqrt{4/5})$, respectively. The difference in their desired function values and their squared metric distance are
$$
 \begin{aligned}
 f(\boldsymbol m_R)-f(\boldsymbol m_L)
 &=\int_0^{1/2}h(t)\,dt+10\int_{9/10}^1r(t)\,dt
 =\frac{11}{40}+1=\frac{51}{40},\\
 (\boldsymbol m_R-\boldsymbol m_L)^\T G
 (\boldsymbol m_R-\boldsymbol m_L)&=\left(\frac12\right)^2
 +100\left(-\frac1{10}\right)^2=\frac54.
 \end{aligned}
$$
Consequently,
$$
 \sup_{\boldsymbol m\ne\boldsymbol m'}
 \frac{[f(\boldsymbol m)-f(\boldsymbol m')]^2}
 {(\boldsymbol m-\boldsymbol m')^\T G(\boldsymbol m-\boldsymbol m')}
 \ge\frac{2601}{2000}>1.
$$
Equation~\eqref{eq:fixed-minimax-limit} gives the leading worst-case MSE coefficient $1/(2K^2)$, whereas
$\frac{1}{2\Lambda(\enc;f)}\ge\frac{2601}{4000K^2}.$
Thus, this coefficient in Theorem~\ref{thm:identifiable-minimax} can be strictly smaller than $1/[2\Lambda(\enc;f)]$ under the same noise normalization.

The strict gap also occurs within the class of receiver postprocessing functions specified in Theorem~\ref{thm:rigidity}. By polynomial approximation, choose polynomials $h_n,r_n$ converging uniformly to $h,r$ on $[-1,1]$ and $[0,1]$, respectively. Prescribe $\psi_{1,n}(m_1)=\int_0^{m_1} h_n(t)\,dt$ and $\psi_{2,n}(m_2)=-10\int_0^{m_2} r_n(t)\,dt$. Both $h$ and $r$ take the values $0$ and $1$. Approximations $h_n,r_n$ within $1/4$ are therefore nonconstant. Their integrals $\psi_{1,n},\psi_{2,n}$ are nonaffine polynomials on every open subinterval. The desired functions $f_n(\boldsymbol m)=\psi_{1,n}(m_1)+\psi_{2,n}(m_2)$ converge uniformly to $f$, and
$$
 \sup_{\mathcal M_4}\bigl(h_n(m_1)^2+r_n(m_2)^2\bigr)\longrightarrow1,
 \qquad
 \frac{[f_n(\boldsymbol m_R)-f_n(\boldsymbol m_L)]^2}{5/4}
 \longrightarrow\frac{2601}{2000}.
$$
For all sufficiently large $n$, the coefficient in
\eqref{eq:fixed-minimax-limit} remains strictly smaller than
$1/[2\Lambda(\enc;f_n)]$, although $\psi_{1,n},\psi_{2,n}$
satisfy the assumptions of Theorem~\ref{thm:rigidity}.

\section{Proofs for Section~\ref{sec:examples}}\label{app:example-proofs}

\subsection{Proof of Lemma~\ref{prop:score}}\label{app:prof:score}

The moment gradient is $\boldsymbol w(\boldsymbol m)=(-\psi'(m_1),1)^\T$.
Monotonicity of $\psi'$ and convexity of the quadratic form below in $\psi'(m_1)$ give
$$
\boldsymbol w(\boldsymbol m)^\T G^{-1}\boldsymbol w(\boldsymbol m)
\le B:=\max_{j\in\{-,+\}}\boldsymbol w_j^\T G^{-1}\boldsymbol w_j.
$$
Since $f(\boldsymbol m)=m_2-\psi(m_1)$ is defined on
$\mathcal C=\mathbb I\times\R$, the segment-integration argument in
Appendix~\ref{app:gradient-bounds} gives
$|f(\boldsymbol m)-f(\boldsymbol m')|
\le\sqrt B\|\boldsymbol m-\boldsymbol m'\|_G$.
Hence, $\Lambda(\enc;f_{\phi,\psi})\ge K^2/B$.

For sharpness, fix $m_1\in(-1,1)$. Choose distinct nearby $s_1,s_2$ averaging to $m_1$ and set the remaining inputs to $m_1$.
Strict monotonicity of $\phi'$ gives $\phi'(s_1)\ne\phi'(s_2)$, so for any moment direction $\boldsymbol d$ we can solve
$$
\frac1K\begin{pmatrix}1&1\\\phi'(s_1)&\phi'(s_2)\end{pmatrix}
 \begin{pmatrix}\delta_1\\\delta_2\end{pmatrix}=\boldsymbol d.
$$

Perturb these inputs to $s_1+h\delta_1,s_2+h\delta_2$.
For sufficiently small $h$, they remain admissible and give
$\boldsymbol m(h)=\boldsymbol m+h\boldsymbol d+o(h)$.
Choose $\boldsymbol d=G^{-1}\boldsymbol w(\boldsymbol m)$.
Then $f(\boldsymbol m(h))-f(\boldsymbol m)
=h\boldsymbol w(\boldsymbol m)^\T G^{-1}\boldsymbol w(\boldsymbol m)+o(h)$, and
$$
\frac{K^2[\boldsymbol m(h)-\boldsymbol m]^\T G[\boldsymbol m(h)-\boldsymbol m]}
{|f(\boldsymbol m(h))-f(\boldsymbol m)|^2}
\longrightarrow
\frac{K^2}{\boldsymbol w(\boldsymbol m)^\T G^{-1}\boldsymbol w(\boldsymbol m)}.
$$

Let $m_1$ approach the maximizing endpoint in \eqref{eq:exactscore}.
Choose admissible perturbations with received distance below $1/n$ and quotient within $1/n$ of its local limit.
Continuity of $\psi'$ proves equality and the vanishing-distance claim.

For $G\succ0$ and $\lambda\ge0$, the inequality
$G\succeq\lambda\boldsymbol w\boldsymbol w^\T$
is equivalent to
$\lambda\boldsymbol w^\T G^{-1}\boldsymbol w\le1$.
Thus, $\Lambda(\enc;f_{\phi,\psi})\ge K^2\lambda$ is equivalent to
$G\succeq\lambda\boldsymbol w_j\boldsymbol w_j^\T$, $j\in\{-,+\}$.
If a singular $G\succeq0$ satisfies the pairwise constraints
in \eqref{eq:general-secant-program}, then
$G+\varepsilon I_2$ also satisfies them for every
$\varepsilon>0$.
The positive-definite case followed by $\varepsilon\downarrow0$ gives the same matrix inequalities.
Conversely, for $\lambda>0$, these inequalities give
$0=\boldsymbol z^\T G\boldsymbol z\ge\lambda(\boldsymbol w_j^\T\boldsymbol z)^2$ for any $\boldsymbol z\in\ker G$.
Strict monotonicity of $\psi'$ makes $\boldsymbol w_-,\boldsymbol w_+$ linearly independent, forcing $\boldsymbol z=0$ and hence $G\succ0$.
At $\lambda=0$, both sets of constraints hold for every $G\succeq0$.

\subsection{Proof of Proposition~\ref{thm:parabola}}\label{app:quartic-proof}

Theorem~\ref{thm:rigidity} and Section~\ref{sec:examples} reduce both cases to a shared encoder $\enc$. Write
$G=\left(\begin{smallmatrix}A&C\\C&B\end{smallmatrix}\right)$,
where $A=|u|^2$, $B=|v|^2$, and $C=\Re(\overline u v)$.

\subsubsection{Quadratic Encoders}\label{app:quadratic-proof}

Encoders that violate \eqref{eq:exactness} have zero score. For identifiable encoders $\enc$,
evenness of $\psi$ and Lemma~\ref{prop:score} give
$$
 \frac{{\Lambda(\enc;f_{\chi_2,\psi})}}{K^2}
 =\frac{AB-C^2}{A+M^2B+2M|C|}.
$$
Normalize $P=1$. The parameter $b\in(1,2)$ in \eqref{eq:parabola} satisfies $M^2b(b-1)=(2-b)^2$.
Averaging \eqref{eq:peak-model} at $\pm1$ gives $A+|c+v|^2\le1$; at zero, $|c|^2\le1$.
Set $S=bA+(b-1)B$ and $\Delta_M=A+M^2B$. Then,  
\begin{align*}
 { S}+|bc+v|^2
 =b\bigl(A+|c+v|^2\bigr)+b(b-1)|c|^2\le b^2,
\end{align*}
which implies $S\le b^2$. The parameter equation then yields the certificate
\begin{align*}
 b(b-1)\Delta_M { S}-bAB
 &=(b-1)\bigl[bA-(2-b)B\bigr]^2\ge0.
\end{align*}
Combining the certificate with $S\le b^2$ gives
$$
\frac{\Lambda(\enc;f_{\chi_2,\psi})}{K^2}
=\frac{AB-C^2}{\Delta_M+2M|C|}
\le\frac{AB}{\Delta_M}
\le(b-1)S\le b^2(b-1).
$$

The proposed encoder has $A=b(2-b)$, $B=b^2$, and $C=0$, attaining equality. It satisfies
$|\enc_{\chi_2}^\star(s)|^2=1-b^2s^2(1-s^2)\le1$ for $s\in\mathbb I$.
Scaling by $\sqrt P$ and solving the parameter equation for $b\in(1,2)$ give Proposition~\ref{thm:parabola}.

\subsubsection{Quartic Encoders}\label{app:quartic-case-proof}

Equation~\eqref{eq:quarticroot} has a unique root $\rho\in(0,1/3)$: its left side increases from zero, while its right side decreases to zero.

For the quartic encoder $\enc(s)=c+us+vs^4$, averaging the peak constraints at $s$ and $-s$ gives $As^2+|c+vs^4|^2\le P$.
Set $b=\sqrt B$ and $c_0=\Re(\overline v c)/b$. The orthogonal encoder $\enc_\perp(s)=\sqrt A\,s+i(bs^4+c_0)$ satisfies
$$
|\enc_\perp(s)|^2=As^2+(bs^4+c_0)^2
\le As^2+|c+vs^4|^2\le P.
$$
Its Gram matrix retains $A,B$ and sets $C=0$, so Lemma~\ref{prop:score} gives $\Lambda(\enc_\perp;f_{\chi_4,\psi})\ge\Lambda(\enc;f_{\chi_4,\psi})$.
Normalize $P=1$ and put $\xi=s^2$. We therefore minimize the convex objective $F(A,b,c_0):=M^2/A+1/b^2$ subject to
$g_\xi(A,b,c_0):=A\xi+(b\xi^2+c_0)^2-1\le0$ for $0\le\xi\le1$, with $A,b>0$.
The proposed coefficients $(A,b,c_0)=(a_\rho^2,b_\rho,c_\rho)$, with $B=b_\rho^2$, obey the exact identity
\begin{equation}
 1-\left[A\xi+(b_\rho\xi^2+c_\rho)^2\right]
 =B(1-\xi)(\xi-\rho)^2(\xi+1+2\rho).
 \label{eq:quarticpeak}
\end{equation}
This verifies all peak constraints. Equality holds at $\xi=\rho$ and $\xi=1$.

For global optimality, retain just these two active convex constraints and set
$$
 \eta=\frac{2}{B^2(1+\rho)^3(1-\rho)},\qquad
 \zeta=\eta\frac{1+\rho}{1-3\rho}.
$$
Both multipliers are positive. Let $h_\rho=b_\rho\rho^2+c_\rho=-b_\rho(1+\rho)^2/2$ and $h_1=b_\rho+c_\rho=b_\rho(1+\rho)(1-3\rho)/2$. The three stationarity identities are
\begin{subequations}
\begin{align}
 -M^2/A^2+\eta\rho+\zeta&=0,\\
 -2/b_\rho^3+2\eta h_\rho\rho^2+2\zeta h_1&=0,\\
 2\eta h_\rho+2\zeta h_1&=0.
\end{align}
\label{eq:kkt}
\end{subequations}
The last two identities in \eqref{eq:kkt} follow by substitution. Using $A=2B\rho(1+\rho)^2$, the first reduces to \eqref{eq:quarticroot}.
Define the convex Lagrangian $\mathcal L=F+\eta g_\rho+\zeta g_1$.
Equation~\eqref{eq:kkt} gives $\nabla\mathcal L(\boldsymbol x_\star)=0$ at $\boldsymbol x_\star=(a_\rho^2,b_\rho,c_\rho)$.
For every feasible $\boldsymbol x=(A,b,c_0)$,
$$
F(\boldsymbol x)\ge\mathcal L(\boldsymbol x)
\ge\mathcal L(\boldsymbol x_\star)=F(\boldsymbol x_\star),
$$
where the last equality uses $g_\rho(\boldsymbol x_\star)=g_1(\boldsymbol x_\star)=0$.
Together with \eqref{eq:quarticpeak}, this proves global optimality. Scaling restores $P$ and gives \eqref{eq:quarticscore}.

% \subsubsection{Baseline Feasibility}\label{app:baseline-feasibility}
% The quadratic baseline in Figure~\ref{fig:curves} satisfies \eqref{eq:peak-model} by the quadratic
% peak identity with $b=1$. For the quartic baseline, setting
% $x=s^2\in[0,1]$ gives
% $$
% \frac{|\enc_{\rm base}(s)|^2}{P}
% =\frac45\left[x+\left(x^2-\frac12\right)^2\right]\le1,
% $$
% because $x\le1$ and $(x^2-1/2)^2\le1/4$.
% For $\psi(s)=s^2$ and $\psi(s)=s^4$, respectively,
% \eqref{eq:exactscore} gives the baseline scores $K^2P/5$ and
% $4K^2P/85$ reported in Table~\ref{tab:values}.

\subsection{Certified Design}\label{sec:certified}

\begin{theorem}[Certified discretization]
\label{thm:certifiedgrid}
Let $K\ge1$, $P>0$, and $\phi_{1,k},\phi_{2,k}\in C^2(\mathbb I)$
for every node $k$.
Assume that $1,\phi_{1,k},\phi_{2,k}$ are linearly independent for each $k$,
and that $f$ is continuous and nonconstant on $\mathcal M_K$.
Let $\lambda_\star$ be the optimum of the exact secant program \eqref{eq:general-secant-program}.
For each node, choose finite grids $\mathcal T_{k,h}$ containing both endpoints and three
fixed inputs whose vectors $(1,\phi_{1,k}(s),\phi_{2,k}(s))^\T$ span $\mathbb R^3$;
each grid has maximum adjacent spacing at most $h$.
Replace only the peak constraints in
\eqref{eq:general-secant-program} by those at $\mathcal T_{k,h}$, retain
every pairwise received-separation constraint, and denote an optimizer by
$(H_{1,h},\ldots,H_{K,h},{ \lambda_h})$.
Write $G_h=(H_{1,h})_{2:3,2:3}$ for its common block.
Such an optimizer exists. Define
$$
\begin{aligned}
\kappa_{0,k}&=\max_{s\in\mathbb I}
\sqrt{1+\phi_{1,k}(s)^2+\phi_{2,k}(s)^2},\\
\kappa_{j,k}&=\max_{s\in\mathbb I}
\sqrt{\phi_{1,k}^{(j)}(s)^2+\phi_{2,k}^{(j)}(s)^2}
\quad(j=1,2).
\end{aligned}
$$
Set
$$
M_h=\max_{1\le k\le K}2\|H_{k,h}\|_2
(\kappa_{1,k}^2+\kappa_{0,k}\kappa_{2,k}),\qquad
\rho_h=\frac{P}{P+M_hh^2/8}.
$$
Then $(\rho_hH_{1,h},\ldots,\rho_hH_{K,h},\rho_h{ \lambda_h})$ is feasible for
\eqref{eq:general-secant-program}, and
\begin{equation}
\rho_h{ \lambda_h}\le\lambda_\star\le { \lambda_h},
\qquad { \lambda_h}-\rho_h{ \lambda_h}=O(h^2)\quad(h\downarrow0).
\label{eq:certifiedgrid}
\end{equation}
The reconstruction in Section~\ref{sec:synthesis} yields peak-feasible
encoders $\enc$ with robustness score $\Lambda(\enc;f)\ge K^2\rho_h\lambda_h$.
Under the assumptions of Theorem~\ref{thm:rigidity}, the bounds
also certify the optimum over all continuous peak-feasible encoders.

\end{theorem}

\begin{proof}

Let $Z_k$ contain the three fixed spanning vectors as columns.
The grid constraints give $Z_k^\T H_kZ_k\succeq0$ and
$\operatorname{tr}(Z_k^\T H_kZ_k)\le3P$. Hence,
$$
\|H_k\|_2\le\|Z_k^{-1}\|_2^2\|Z_k^\T H_kZ_k\|_2
\le\|Z_k^{-1}\|_2^2\operatorname{tr}(Z_k^\T H_kZ_k)
\le3P\|Z_k^{-1}\|_2^2.
$$
Let $C_H$ be the largest of these bounds.
Since $f$ is nonconstant, choose attainable $\boldsymbol m_0,\boldsymbol m_0'$ with
$a_0=|f(\boldsymbol m_0)-f(\boldsymbol m_0')|^2>0$.
For $\boldsymbol d_0=\boldsymbol m_0-\boldsymbol m_0'$, the separation constraint gives
$0\le\lambda\le\boldsymbol d_0^\T G\boldsymbol d_0/a_0
\le C_H\|\boldsymbol d_0\|_2^2/a_0$.
These bounds are independent of $h$. Both feasible sets contain zero
and are closed and bounded, so their maxima are attained.

With $\boldsymbol q_k(s)=(1,\phi_{1,k}(s),\phi_{2,k}(s))^\T$,
write $q_{H_k}(s)=\boldsymbol q_k(s)^\T H_k\boldsymbol q_k(s)$.
Differentiating twice gives
$$
\begin{aligned}
q_{H_k}''(s)
&=2\boldsymbol q_k'(s)^\T H_k\boldsymbol q_k'(s)
+2\boldsymbol q_k(s)^\T H_k\boldsymbol q_k''(s),\\
|q_{H_{k,h}}''(s)|
&\le2\|H_{k,h}\|_2(\kappa_{1,k}^2+\kappa_{0,k}\kappa_{2,k})
\le M_h.
\end{aligned}
$$
For adjacent grid inputs $t<t'$ and $s\in[t,t']$, the linear-interpolation error bound gives
\begin{align*}
q_{H_{k,h}}(s)
\le\frac{t'-s}{t'-t}q_{H_{k,h}}(t)
 +\frac{s-t}{t'-t}q_{H_{k,h}}(t')
 +\frac{M_h}{2}(s-t)(t'-s)
\le P+M_hh^2/8.
\end{align*}
Scaling $H_{k,h}$ and $\lambda_h$ by $\rho_h$ gives $q_{\rho_hH_{k,h}}(s)=\rho_hq_{H_{k,h}}(s)\le P$
and preserves separation, so $\rho_h\lambda_h\le\lambda_\star$.
Relaxing the peak constraints gives $\lambda_\star\le\lambda_h$.
The uniform bounds give
$\lambda_h(1-\rho_h)=\lambda_h M_hh^2/(8P+M_hh^2)=O(h^2)$. Finally, we note that the reconstruction in Appendix~\ref{app:general-robust} preserves
$\rho_hG_h$ without increasing power, proving the encoder claim.

\end{proof}

Peak rescaling by $\rho_h$ does not certify omitted received-separation constraints in \eqref{eq:general-secant-program}.
For identical local preprocessing functions and a common grid, replacing
all $H_{k,h}$ by $H_{1,h}$ preserves feasibility and the objective.

\subsection{Verified Primal and Dual Bounds}\label{app:certificate-witness-bounds}

Let $\lambda\ge0$ and let $(H_1,\ldots,H_K,\lambda)$ satisfy the positive-semidefinite,
common-block, and pairwise received-separation constraints in \eqref{eq:general-secant-program}. Suppose $b_h$ bounds
all grid values of $q_{H_k}$, and $M$ bounds $|q_{H_k}''|$ on
$\mathbb I$ for every node. For $b_h+Mh^2/8>0$, scaling the
candidate gives the certified lower bound
$$
\lambda_{\mathrm{lo}}=\frac{P\lambda}{b_h+Mh^2/8}\le\lambda_\star.
$$
For a finite dual certificate, choose pairs of attainable moment vectors
$(\boldsymbol m_r,\boldsymbol m_r^{\prime})$, write
$\boldsymbol d_r=\boldsymbol m_r-\boldsymbol m_r^{\prime}$ and
$a_r=(f(\boldsymbol m_r)-f(\boldsymbol m_r^{\prime}))^2$, and choose
nonnegative weights $\alpha_{k,i},\beta_r$ and symmetric $2\times2$ matrices
$W_k$ satisfying
\begin{align}
\nonumber
Q_k=\sum_i\alpha_{k,i}
\begin{pmatrix}1\\\phi_{1,k}(t_{k,i})\\\phi_{2,k}(t_{k,i})\end{pmatrix}
\begin{pmatrix}1&\phi_{1,k}(t_{k,i})&\phi_{2,k}(t_{k,i})\end{pmatrix}-\operatorname{diag}(0,W_k)&\succeq0
\quad(k=1,\ldots,K),\\
\sum_{k=1}^K W_k=\sum_r\beta_r\boldsymbol d_r\boldsymbol d_r^\T,
\qquad \sum_r\beta_r a_r\ge1,
\qquad t_{k,i}\in\mathbb I.&
\label{eq:general-dualwitness}
\end{align}
Then ${ \lambda_{\mathrm{up}}}=P\sum_{k,i}\alpha_{k,i}$ is a certified upper bound.
The absolute gap is at most $\lambda_{\mathrm{up}}-\lambda_{\mathrm{lo}}$.
If $\lambda_{\mathrm{lo}}>0$, the rescaled encoders' relative loss from optimal robustness is at most
$(\lambda_{\mathrm{up}}-\lambda_{\mathrm{lo}})/\lambda_{\mathrm{up}}$.
The sufficient conditions \eqref{eq:general-dualwitness} guarantee neither dual attainment by a finite certificate nor arbitrarily small gaps from sampling pairs of moment vectors.

\begin{proof}[Proof of the certificate bounds]

The same interpolation argument gives $q_{H_k}(s)\le b_h+Mh^2/8$.
Scaling all $H_k$ and $\lambda$ by $P/(b_h+Mh^2/8)$ therefore
proves $\lambda_{\mathrm{lo}}\le\lambda_\star$
(Appendix~\ref{app:general-robust} reconstructs the corresponding encoders).  For any feasible $(H_1,\ldots,H_K,\lambda)$, \eqref{eq:general-dualwitness} gives
$\operatorname{tr}(Q_kH_k)
=\sum_i\alpha_{k,i}q_{H_k}(t_{k,i})-\operatorname{tr}(W_kG)\ge0$,
since $Q_k,H_k\succeq0$.
The peak and separation constraints, common block, and normalization
$\sum_r\beta_ra_r\ge1$ now give
\begin{align}
P\sum_{k,i}\alpha_{k,i}
\ge\sum_{k,i}\alpha_{k,i}q_{H_k}(t_{k,i})\ge\sum_k\operatorname{tr}(W_kG)
 =\sum_r\beta_r\boldsymbol d_r^\T G\boldsymbol d_r\ge\lambda\sum_r\beta_r a_r\ge\lambda.
\end{align}
Taking the optimum gives $\lambda_\star\le\lambda_{\mathrm{up}}$.
For $\lambda_{\mathrm{lo}}>0$,
$1-\lambda_{\mathrm{lo}}/\lambda_\star
\le1-\lambda_{\mathrm{lo}}/\lambda_{\mathrm{up}}$,
proving the relative-loss bound.

\end{proof}

The endpoint reduction in Lemma~\ref{prop:score} gives a finite
dual certificate for the examples. Choose finitely many inputs
$t_i\in\mathbb I$, weights $\alpha_i\ge0$, and matrices
$W_-,W_+\in\mathbb S_+^2$ satisfying
\begin{equation}
\begin{aligned}
Q&=\sum_i\alpha_i(1,t_i,\phi(t_i))^\T(1,t_i,\phi(t_i))
-\operatorname{diag}(0,W_-+W_+)\succeq0,\\
&\boldsymbol w_-^\T W_-\boldsymbol w_-
+\boldsymbol w_+^\T W_+\boldsymbol w_+\ge1.
\end{aligned}
\label{eq:dualwitness}
\end{equation}
Then $\lambda_{\mathrm{up}}=P\sum_i\alpha_i$ is an upper bound on $\lambda_\star$.
The gap bounds based on $\lambda_{\mathrm{lo}}$ and $\lambda_{\mathrm{up}}$ in Appendix~\ref{app:certificate-witness-bounds} apply.
Theorem~\ref{thm:certifiedgrid} establishes convergence for
exact grid solutions; numerical certification also requires
the primal and dual verification errors to decrease. 
After the shared-encoder reduction in
Section~\ref{sec:examples}, we apply the
preceding trace argument to the endpoint constraints
$G\succeq\lambda\boldsymbol w_j\boldsymbol w_j^\T$,
$j\in\{-,+\}$, from Lemma~\ref{prop:score}.
Together with \eqref{eq:dualwitness}, this gives
$\lambda_\star\le P\sum_i\alpha_i$.

\subsection{Projection Decoder in the Examples}\label{app:strip-decoder-proof}

Under the assumptions of Lemma~\ref{prop:score}, we specialize
the projection decoder in Section~\ref{sec:decoding} to
$\mathcal C=[-1,1]\times\R$. For the orthogonal polynomial
encoders, $G\succ0$ and $G_{12}=0$, so the projection objective
separates as
$$
\|\boldsymbol m-\tilde{\boldsymbol m}\|_G^2
=G_{11}(m_1-\tilde m_1)^2
+G_{22}(m_2-\tilde m_2)^2.
$$
We minimize this objective over $\mathcal C$ by clipping
$\tilde m_1$ to $[-1,1]$ and leaving $\tilde m_2$ unchanged,
which gives \eqref{eq:polynomial-decoder}. Next, we verify that this decoder satisfies the uniform error
guarantee \eqref{eq:uniform} with the optimal
noise-amplification constant $L_\enc$.
Lemma~\ref{prop:score} gives
$B_{\mathcal C}(G)=K^2/\Lambda(\enc;f)$.
Applying the projection bound in
Appendix~\ref{app:general-decoder-proof}, we obtain
$$
\left|
\dec_{\mathrm{strip}}(r_\enc(\boldsymbol s)+z)
-f(\boldsymbol s)
\right|
\le
\frac{|z|}{\sqrt{\Lambda(\enc;f)}}
$$
for every $\boldsymbol s\in\mathbb I^K$ and $z\in\C$.
Thus, the decoder attains
$L_\enc=\Lambda(\enc;f)^{-1/2}$, which
Proposition~\ref{prop:amplification} establishes as the
smallest achievable constant for the fixed encoder collection.

\subsection{Moment Set and Output Range}\label{app:moment-region-proof}

To justify output clipping for the polynomial examples, consider
$\phi(s)=\psi(s)=s^p$, with $p\in\{2,4\}$. Applying Jensen's
inequality and using $|s_k|\le1$, we obtain
$m_1^p\le m_2\le1$, hence $0\le m_2-m_1^p\le1$.
Clipping either estimated function value to $[0,1]$
therefore cannot increase its absolute error. We next characterize the attainable moment set for the variance
example ($p=2$) with $K=4$, as shown in
Figure~\ref{fig:decoder-geometry}. For a fixed $m_1\in[-1,1]$,
equal inputs $s_k=m_1$ attain the lower boundary $m_2=m_1^2$.
To determine the upper boundary, we maximize the convex function
$\sum_k s_k^2/4$ over
$\{\boldsymbol s\in\mathbb I^4:\sum_k s_k=4m_1\}$.
A maximum occurs at a vertex of this set.

At most one coordinate of a vertex lies in $(-1,1)$.
Otherwise, perturbing two interior coordinates in opposite
directions gives distinct feasible vectors
$\boldsymbol s\pm\varepsilon(\boldsymbol e_i-\boldsymbol e_j)$
for sufficiently small $\varepsilon>0$. Their midpoint is
$\boldsymbol s$, contradicting the vertex property.
Here, $\boldsymbol e_i$ denotes the $i$th coordinate vector.
The remaining three coordinates belong to $\{-1,1\}$ and
sum to $j\in\{-3,-1,1,3\}$, giving the upper boundary
$$
m_{2,\max}(m_1)=\frac{3+(4m_1-j)^2}{4},
\qquad |4m_1-j|\le1.
$$
The line segment joining the minimizing and maximizing input
vectors remains feasible and preserves $m_1$. By continuity
of $m_2$ along this segment, every value between the two
boundaries is attainable.

\bibliographystyle{IEEEtran}
\bibliography{bibliography}

\end{document}